%% file: main.tex
\documentclass[11pt]{article}

\usepackage{microtype}
\usepackage{graphicx}
\usepackage{subfigure}
\usepackage[utf8]{inputenc}
\usepackage{booktabs} 

\usepackage{hyperref}

\input{preamble}
\usepackage{fullpage}

\usepackage{amsmath}
\usepackage{amssymb}
\usepackage{mathtools}
\usepackage{amsthm}

\usepackage[capitalize,noabbrev]{cleveref}

\usepackage[textsize=tiny]{todonotes}

\usepackage{content/StyleSheets/template}
\usepackage{content/StyleSheets/moreStyle}
\usepackage{content/StyleSheets/global}
\usepackage{subfiles}
\usepackage{url,color}
\usepackage{footnote}
\usepackage{threeparttable}
\usepackage{mdframed}
\usepackage{xspace}
\usepackage{pifont}
\usepackage{stmaryrd}
\makesavenoteenv{tabular}
\makesavenoteenv{table}
\newtheorem*{theorem*}{Theorem}

\usepackage{booktabs}
\usepackage{tablefootnote}

\usepackage{array}
\newcolumntype{L}[1]{>{\raggedright\arraybackslash}p{#1}} 
\newcolumntype{C}[1]{>{\centering\arraybackslash}p{#1}}    
\newcolumntype{R}[1]{>{\raggedleft\arraybackslash}p{#1}}   

\begin{document}

\newcommand{\sa}{\mathsf{StreamAlg}}
\newcommand{\dpr}{\mathsf{DistProt}}
\newcommand{\cor}{\mathsf{Coordinator}}
\newcommand{\ser}{\mathsf{Server}}
\newcommand{\pkgen}{\mathsf{PublicKeyGen}}
\newcommand{\skgen}{\mathsf{SecretKeyGen}}

\newcommand{\gsw}{\mathsf{GSW}}
\newcommand{\nand}{\mathsf{NAND}}
\newcommand{\pfhe}{\mathsf{PFHE}}
\newcommand{\update}{\mathsf{Update}}
\newcommand{\report}{\mathsf{Report}}
\newcommand{\tikzxmark}{%
\tikz[scale=0.23] {
    \draw[line width=0.7,line cap=round] (0,0) to [bend left=6] (1,1);
    \draw[line width=0.7,line cap=round] (0.2,0.95) to [bend right=3] (0.8,0.05);
}}
\newcommand{\dgaussian}{\mathcal{G}}
\newcommand{\OO}{\tilde{\mathcal{O}}}
\newcommand{\ldec}{\mathsf{LinearDec}}
\newcommand{\ds}{\mathcal{DS}}
\newcommand{\GapSVP}{\mathsf{GapSVP}}
\newcommand{\C}{\mathcal{C}}
\DeclarePairedDelimiter\ceil{\lceil}{\rceil}
\DeclarePairedDelimiter\floor{\lfloor}{\rfloor}

\renewcommand{\subparagraph}[1]{\medskip\noindent\underline{\textit{#1}}}

\title{Fast White-Box Adversarial Streaming Without a Random Oracle}
\author{Ying Feng\thanks{Carnegie Mellon University. 
E-mail: \email{yingfeng@andrew.cmu.edu}}
\and
Aayush Jain\thanks{Carnegie Mellon University. 
E-mail: \email{aayushja@andrew.cmu.edu}}
\and
David P. Woodruff\thanks{Carnegie Mellon University. 
E-mail: \email{dwoodruf@andrew.cmu.edu}}
}
\date{\today}
\maketitle




\input{content/abstract}

\newpage

\input{content/intro}

\input{content/techover.tex}
\input{content/preliminaries.tex}
\input{content/streaming.tex}
\input{content/efficiency.tex}

\input{content/distributed.tex}

\input{content/matrixandtensor}

\nocite{langley00}

\section*{Acknowledgements}
Aayush Jain was supported by a Google faculty research scholar award 2023. David P.\ Woodruff was supported by a Simons Investigator Award and NSF Grant No. CCF-2335412.

\def\shortbib{0}
\bibliographystyle{alpha}

\bibliography{content/Bibliography/abbrev3,content/Bibliography/stream,content/Bibliography/crypto,content/Bibliography/custom}



\end{document}

%% file: preamble.tex
\providecommand{\email}[1]{\href{mailto:#1}{\nolinkurl{#1}\xspace}}

\hypersetup{
     colorlinks   = true,
     citecolor    = cyan,
     linkcolor		= purple
}

%% file: content/abstract.tex
\begin{abstract}
Recently, the question of adversarially robust streaming, where the stream is allowed to depend on the randomness of the streaming algorithm, has gained a lot of attention. In this work, we consider a strong white-box adversarial model (Ajtai et al. PODS 2022), in which the adversary has access to all past random coins and the parameters used by the streaming algorithm. We focus on the sparse recovery problem and extend our result to other tasks such as distinct element estimation and low-rank approximation of matrices and tensors. The main drawback of previous work is that it requires a {\it random oracle}, which is especially problematic in the streaming model since the amount of randomness is counted in the space complexity of a streaming algorithm. Also, the previous work suffers from large update time. We construct a near-optimal solution for the sparse recovery problem in white-box adversarial streams, based on the subexponentially secure Learning with Errors assumption. Importantly, our solution does not require a random oracle and has a polylogarithmic per item processing time. We also give results in a related white-box adversarially robust distributed model. Our constructions are based on homomorphic encryption schemes satisfying very mild structural properties that are currently satisfied by most known schemes.

\end{abstract}

%% file: content/intro.tex
\section{Introduction}\label{sec:intro}

The streaming model of computation has emerged as an increasingly popular paradigm for the analysis of massive datasets, where the overwhelming size of the data places restrictions on the memory, computation time, and other resources available to the algorithm. In the streaming model, the input data is implicitly defined through a stream of data elements that sequentially arrive one-by-one. One can also delete a previous occurrence of an item. An algorithm makes one pass over the stream and uses limited space to approximate a desired function of the input. In the formalization of Alon, Matias, and Szegedy \cite{10.1145/237814.237823}, the data is represented as an underlying $n$-dimensional vector that is initialized to $0^n$, and then undergoes a sequence of additive updates to its coordinates. The algorithm aims to optimize the space and time to process these updates, while being able to respond to queries about this underlying vector.

This model captures key resource requirements of algorithms for many practical settings such as
Internet routers and traffic logs \cite{10.1145/1140103.1140295,10.1145/1071690.1064258,1623822,Venkataraman2005NewSA,1561219,10.1145/1353343.1353442}, databases \cite{SelingerACLP79,DasuJMS02,FLAJOLET1985182,Lightstone18,ShuklaDNR96,BrownHMPRS05}, sensor networks \cite{10.1145/3532189,Gaber2007}, financial transaction data \cite{ABS}, and
scientific data streams \cite{10.1145/872757.872857,SAGJ}. 
See \cite{TCS-002} for an overview of streaming
algorithms and their applications.

A large body of work on streaming algorithms has been designed for {\it oblivious} streams, for which the sequence of updates may be chosen adversarially, but it is chosen independently of the randomness of the streaming algorithm. Recently, there is a growing body of work on robust streaming algorithms in the black-box adversarial model \cite{Ben-EliezerY20,Ben-EliezerJWY21,HassidimKMMS20,WoodruffZ21,AlonBDMNY21,KaplanMNS21,BravermanHMSSZ21,MenuhinN21,AttiasCSS21,Ben-EliezerEO22,ChakrabartiGS22}, in which the adversary can monitor the outputs of the streaming algorithm and choose future stream updates based on these outputs. While useful in a number of applications, there are other settings where the adversary may also have access to the internal state of the algorithm, and this necessitates looking at a stronger adversarial model known as the white-box adversarial streaming model.

\subsection{White-box Adversarial Streaming Model}

We consider the white-box adversarial streaming model introduced in \cite{10.1145/3517804.3526228}, where a sequence of stream updates is chosen adaptively by an adversary who sees the full internal state of the algorithm at all times, including the parameters and the previous randomness used by the algorithm. The interaction between an adversary and the streaming algorithm is informally depicted as the following game. 
\begin{itemize}[label={},leftmargin=*]
    \item {\it Consider a multi-round, two-player game between $\sa$, the streaming algorithm, and an adversary $\A$. Prior to the beginning of the game, fix a query function $f$, which asks for a function of an underlying dataset that will be implicitly defined by the stream. In each round:

    \begin{enumerate}[leftmargin=*]
        \item $\A$ computes an update $x$ for the stream, which depends on all previous stream updates, all previous internal states, and the randomness used by $\sa$.  

        \item $\sa$ acquires a fresh batch of random bits, uses $x$ to update its data structures, and (if asked) outputs a response to the query function $f$.

        \item $\A$ observes the random bits, the internal state of $\sa$, and the response.
    \end{enumerate}
    
    The goal of $\A$ is to make $\sa$ output an incorrect query at some round during the game.}
\end{itemize}

The white-box adversarial model captures characteristics of many real-world attacks, where an adaptive adversary has access to the entirety of the internal states of the system. In comparison to the oblivious stream model or the black-box adversarial model \cite{Ben-EliezerJWY21}, this model allows for richer adversarial scenarios.

For example, consider a distributed setting where a centralized server collects statistics of a database generated by remote users. The server may send components of its internal state $S$ to the remote users in order to optimize the total communication over the network. The remote users may use $S$ in some process that generates downstream data. Thus, future inputs depend on the internal state $S$ of the algorithm run by the central coordinator. In such settings, the white-box robustness of the algorithms is crucial for optimal selection of query plans~\cite{SelingerACLP79}, online analytical processing~\cite{ShuklaDNR96}, data integration~\cite{BrownHMPRS05}, and data warehousing~\cite{DasuJMS02}.

Many persistent data structures provide the ability to quickly access previous versions of information stored in a repository shared across multiple collaborators. The internal persistent data structures used to provide version control may be accessible and thus visible to all users of the repository. These users may then update the persistent data structure in a manner that is not independent of previous states~\cite{DriscollSST89,FiatK03,Kaplan04}.

Dynamic algorithms often consider an adaptive adversary who generates the updates upon seeing the entire data structure maintained by the algorithm during the execution~\cite{Chan10,ChanH21,RoghaniSW22}. For example, \cite{Wajc20} assumes the entire state of the algorithm (including the set of randomness) is available to the adversary after each update, i.e., a white-box model.

Moreover, robust algorithms and adversarial attacks are important topics in machine learning~\cite{SzegedyZSBEGF14,GoodfellowSS14}, with a large body of recent literature focusing on adversarial robustness of machine learning models against white-box attacks~\cite{IlyasEM18,MadryMSTV18,SchmidtSTTM18,TramerKPGBM18,CubukZMVL18,KurakinGB17,LiuCLS17}.

\subsection{Prior Work in the White-Box Model}

We aim to design {\it white-box adversarially robust} (WAR) streaming algorithms with provable correctness guarantees against such powerful adversaries. In general, it is difficult to design non-trivial WAR algorithms,
due to the fact that many widely used space-saving techniques are subject to
attacks when all parameters and generated randomness of the system are immediately revealed to the adversary. 
In the streaming community, both \cite{10.1145/3517804.3526228} and \cite{fw23} explicitly study the white-box model for data streams and suggest the use of cryptography in designing WAR algorithms (see also \cite{Ben-EliezerJWY21} for the use of cryptography for black box robustness). A similar model to the white-box model has also been independently developed in the cryptography community in the study of property preserving hashing \cite{BLV19}.

In this work, the central problem we consider is $k$-sparse recovery as follows:
\begin{center}
{\bf Sparse Recovery:} 
given an input vector $\vec{x}$, if $\vec{x}$ contains at most $k$ non-zero entries, recover $\vec{x}$. Otherwise, report \emph{invalid input}.
\end{center}
A primitive for sparse recovery in the white-box model can be used for tasks such as distinct element estimation, low rank approximation of matrices and tensors, and finding a maximum matching \cite{10.1145/3517804.3526228,fw23}. Sparse recovery is also itself a critical data acquisition and processing problem that arises in signal and image processing, machine learning, data networking, and medicine, see \cite{10.1145/3517804.3526228,fw23} and the references therein. 

We note that a weaker definition of sparse recovery assumes that the input vector $\vec{x}$ is promised to be $k$-sparse, and only under this promise recovers $\vec{x}$. If the input is not $k$-sparse, the algorithm can output any vector. This problem can in fact be solved deterministically (and thus also in the white-box model) with low memory (see, e.g., \cite{4797556}, for a version with fast update time, though other simpler schemes based on MDS codes exist). 
The drawback under this weaker definition though is that, without knowing the sparsity of the input ahead of time, a user of the algorithm cannot tell whether the output is a correct recovery or garbage. In the latter case, the client might ignorantly use the garbage output in subsequent applications and propagate such an error.

\cite{fw23} solves this problem based on the assumed hardness of the Short-Integer-Solution (SIS) problem, which is a problem that is studied in lattice-based cryptography. The robust version of the property preserving hashing \cite{BLV19} can also be applied to solve the same problem. However, a notable drawback of all prior solutions proposed is their reliance on either a {\it random oracle} or a prohibitively long random string. This is because a key component of these algorithms is a giant, random matrix $\mat{A}$ for which it is hard to find a vector $\vec{x}$ with small entries for which $\mat{A}\vec{x}$ is sparse. Long random strings are especially problematic in data streams, where the resource measure is the total space complexity, including the randomness, and there is a large body of work on derandomizing such algorithms, see, e.g.,  \cite{Nisan92,Indyk06,KNW10,kacham2023pseudorandom}. Since it was previously unknown how to generate $\mat{A}$ pseudorandomly while maintaining this hardness property, previous work \cite{fw23} relies on a random oracle to heuristically compress such a matrix $\mat{A}$. This assumes that the algorithm is given read access to a long string of random bits, which is often implemented with hash-based heuristic functions such as AES or SHA256. Unfortunately, this heuristic cannot be proven secure using a standard reduction-style security proof. This motivates a central question of this work: Do there exist efficient streaming algorithms in the white-box model without a random oracle?

\subsection{Our Results}
Our algorithms make use of the hardness assumptions of the
Learning with Error (LWE) problem and its variant, Ring LWE. Both problems are popular in cryptography known for building homomorphic encryption schemes. In this section, we state the efficiency of our algorithms under the sub-exponential hardness assumptions of these problems. However, we stress that in contrast to \cite{fw23}, our algorithms are robust even under weaker polynomial hardness assumptions.
We let $n$ be the dimension of the underlying input vector, and for notational convenience, we assume that the entries of the vector are integers bounded in absolute value by $\poly(n)$ at all times during the stream. We use $\tilde{\mathcal{O}}(f)$ to denote $f \cdot \poly(\log n)$. 

Our main result is informally stated as follows:

\begin{theorem}[Informal]
    Assuming the sub-exponential hardness\footnote{Sub-exponential hardness means that we assume that for adversaries of size $2^{\secparam^\beta}$,  where $\secparam$ is the security parameter and $\beta > 0$ is some constant, a sample from an LWE distribution is indistinguishable from a uniformly random sample. 
    } of the Learning with Errors (LWE) problem, there is a WAR streaming algorithm for $k$-sparse recovery, which:
    \begin{itemize}[itemsep=0em, label=\textperiodcentered]
        \item takes $\OO(k)$ bits of space,

        \item has $\OO(1)$ update time, and
        
        \item has $\OO(k^{1+c})$ report time for an arbitrarily small constant $c > 0$.
        
    \end{itemize}
\end{theorem}

This matches (up to a poly-logarithmic factor) the optimal space and update time complexity for the weaker $k$-sparse recovery problem under the easier oblivious streaming setting. On the other hand, assuming polynomial security of the LWE problem, there is a multiplicative overhead of $n^{\epsilon}$ for arbitrary constant $\epsilon>0$ on our space and time complexities.

\paragraph{Adversarial Model in the Distributed Setting.} We introduce and formalize the white-box adversarial model for distributed data processing as a means of modeling the existence of one or more malicious parties in a distributed system. We also study the sparse recovery problem in a white-box distributed model. A WAR streaming algorithm can be naively converted to a WAR distributed algorithm under certain conditions, yet the resulting processing time of the servers is $\calO(n)$ times that of the streaming update time.  In contrast to this suboptimal time implied by the naive conversion, we show that we can achieve a near-optimal processing time in a white-box distributed model assuming the {\it polynomial} hardness of Ring-LWE, which is a variant of the LWE problem defined over a polynomial ring. This is because Ring-LWE supports SIMD (Single-Instruction Multiple-Data) style operations that have been exploited in the design of efficient homomorphic encryption schemes in cryptography. We informally state our result as follows:

\begin{theorem}[Informal]
    Assuming the hardness of Ring-LWE, there exists a WAR distributed protocol for $k$-sparse recovery, which:
    \begin{itemize}[itemsep=0em, label=\textperiodcentered]
        \item uses $\OO(k)$ bits of communication between each server and the coordinator,

        \item has $\OO(n)$ processing time on each server, and
        
        \item the coordinator spends $\OO(\max(n, k^{1+c}))$ time to output the solution given the messages from the servers, for an arbitrarily small constant $c > 0$.
        
    \end{itemize}
\end{theorem}

We summarize the complexities of our sparse recovery algorithms in Table \ref{tab:summary}, as compared to the best-known upper bounds for these problems in the white-box adversarial model.

\input{content/table1}

\paragraph{Extensions to matrix and tensor recoveries.}
Similar to \cite{fw23}, generalizing our construction produces low-rank matrix and tensor recovery algorithms, again with the crucial property that they do not assume a random oracle or store a large random seed. 

Consider a data stream updating the entries of an underlying matrix or tensor, which can be represented as its vectorization in a stream. Let $n \geq m$ be the dimensions of an underlying input matrix, and let $n$ be the dimension parameter of an underlying data tensor in $\R^{n^d}$ for some $d \in \calO(1)$ (we assume all dimensions of the tensor are the same only for ease of notation). We have the following results in the streaming model:

\begin{theorem}[Informal]
    Assuming the sub-exponential hardness of the Learning with Errors (LWE) problem, there is a WAR streaming algorithm for rank-$k$ matrix recovery, which takes $\OO(nk)$ bits of space. There is also a WAR streaming algorithm for rank-$k$ tensor recovery, which takes $\OO(nk^{\ceil{\log d}})$ bits of space.
\end{theorem}

We remark that the above matrix recovery algorithm performs $\OO(nk)$ measurements using $n$-sparse matrices, which can be converted into a distributed protocol with $\OO(n^2k)$ processing time. This improves upon the matrix recovery algorithm in the previous work \cite{fw23} which, if implemented in a distributed setting, would require $\OO(n^3k)$ processing time at each server.

\paragraph{Additional Applications.} 
Our $k$-sparse recovery algorithm can be used to construct an $L_0$-norm estimation algorithm following the same recipe in \cite{fw23}. In short, by running our $k$-sparse recovery algorithm with the sparsity parameter $k$ set to $n^{1-\alpha}$ for some constant $\alpha$, we can get an $n^\alpha$-approximation to the $L_0$ norm of a vector. In comparison to \cite{fw23}, our construction again removes the need for a random oracle or the need to store a lot of randomness.

%
%


%% file: content/table1.tex
\tabcolsep=0.15cm
\begin{table*}[h!]
    \caption{A summary of the bit complexities and runtime of our $k$-sparse recovery algorithm, as compared to the best-known upper bounds for these problems in the white-box adversarial model \cite{fw23}. The SIS assumption used in the previous work is proven to be equivalently hard  to the LWE assumption we use, up to different parameters.
    The TIME columns refer to the update time in the streaming model and processing time in the distributed model. $c$ denotes an arbitrarily small positive constant.}\label{tab:summary}
    \begin{center}
    \begin{small}
    \begin{sc}
    \begin{tabular}{lcccccr}
    \toprule
    &  Model & Assumption & Space & Time & Rand. Oracle?\\
    \midrule

     \multirow{2}{*}{\textbf{Prev.}} & Streaming  & subexp SIS & $\OO(k)$ & $\OO(k)$ &  $\checkmark$\\

     & Distributed  & subexp SIS & $\OO(k)$ & $\OO(nk)$ & $\checkmark$ \\

     \addlinespace\hline\addlinespace

     \multirow{5}{*}{\textbf{Ours}} & Streaming & subexp LWE & $\OO(k)$ & $\OO(1)$ & $\tikzxmark$\\

     & Streaming  & poly LWE & $\OO(n^c +k)$ & $\OO(n^c)$ & $\tikzxmark$\\

     & Distributed  & subexp LWE & $\OO(k)$ & $\OO(n)$ & $\tikzxmark$ \\

     & Distributed  & poly LWE & $\OO(n^c +k)$ & $\OO(n^{1+c})$ & $\tikzxmark$ \\

     & Distributed  & poly Ring-LWE & $\OO(n^c +k)$ & $\OO(n)$ & $\tikzxmark$ \\
    \bottomrule
    \end{tabular}
    \end{sc}
    \end{small}
    \end{center}
    \vskip -0.1in
\end{table*}

%% file: content/techover.tex
\subsection{Technical Overview}

We leverage structural properties of lattice-based assumption to address all drawbacks listed above. In particular, departing from the previous construction of \cite{fw23}, we work with a suitable fully homomorphic encryption (FHE) scheme to achieve the desired property. We start with an intuitive explanation of the previous construction in \cite{fw23}. 

\paragraph{Previous Construction} One general framework for solving the sparse recovery problem is to run a weak recovery algorithm as a black-box subroutine, which potentially provides a garbage output, in addition to an additional ``tester'' scheme to verify whether the recovery output matches the original input (i.e., is a correct recovery) or is garbage. In the oblivious streaming setting, such a ``tester'' can be efficiently implemented using, e.g., a random polynomial “fingerprint” \cite{10.1145/3297715}. However, this approach is not WAR as the randomness of the sampler is revealed and thus not independent of the stream. 

Instead, \cite{fw23} uses a random matrix $\mat{A}$ to implement a WAR tester, which is assumed to be heuristically compressed by an idealized hash function. During the stream, \cite{fw23} maintains a hash $h =\mat{A}\vec{x}$, where $\vec{x}$ is the data vector implicitly defined by the stream. In the testing stage, it acknowledges an output $\vec{x'}$ to be a correct recovery if and only if $\vec{x'}$ is $k$-sparse and $\mat{A}\vec{x'}$ equals the hash $h$. The idea is that if there exists a $k$-sparse collision $\vec{x'}$ satisfying $\mat{A}(\vec{x} - \vec{x'}) = 0$ and $\vec{x'} \neq \vec{x}$, then an adversary can brute-force over all $k$-sparse vectors to find such $\vec{x} - \vec{x'}$, which is a solution to the SIS problem. To spawn $n$ matrix columns and support the brute-force reduction, the security parameter (i.e., the number of rows of $\mat{A}$) needs to be in $\omega(k\log n)$, assuming the sub-exponential security of SIS.

\paragraph{Efficient Reduction.} The number of matrix rows in the previous construction depends on the sparsity parameter $k$ due to a brute-forcing step in the reduction. To grant enough time for iterating through all $k$-sparse vectors (with $\poly(n)$-bounded entries), the security parameter has to be greater than $k\log n$, which results in $\OO(k)$ time to process every single update. Our first idea is to construct a poly-time reduction in order to remove such a dependency on $k$. This allows us to achieve better space and time efficiency, in addition to basing the construction on a polynomial hardness assumption.

The observation is that the non-existence of a sparse collision is an overly strong condition. Instead, for our purpose, we only need to ensure that such a collision, if it exists, cannot be found by the poly-time weak recovery scheme that the algorithm runs. Thus, in the reduction, instead of brute-forcing over all $k$-sparse vectors, we let the adversary run the same poly-time weak recovery scheme that the algorithm does. If a collision $\vec{x'}$ satisfying $\mat{A}\vec{x} = \mat{A}\vec{x'}$ is output by the scheme, then the adversary can still break the security assumption.

\paragraph{Reducing the Amount of Randomness.} The rest of our effort will be attributed to reducing the size of the state of the algorithm. Rather than using a truly random giant matrix to hash the data as in the previous construction, we would like to design a short digest (ideally linear-sized in the security parameter) and generate the columns of the hash matrix on the fly. This naturally yields correlated columns as opposed to the truly random columns in an SIS matrix. We aim to show that in this setting, the chance of finding a collision is still negligible.

\paragraph{Streaming Friendly Collision-Resistant Hash.} Our main insight from cryptography to address the issue above is that we could construct a hash function family $\mathcal{H}$, with special properties to help solve our streaming problem:
\begin{itemize}
\item The hash key $h\leftarrow\mathcal{H}$ is small in size, ideally polynomial in the security parameter $\lambda$ (or even polylogarithmic assuming subexponential time assumptions).
\item The hash should be computable in a streaming fashion and it should  be updatable where each update takes polynomial time in the security parameter $\lambda$ (or even polylogarithmic time assuming subexponential time assumptions).
\item The hash key should be pseudorandom. In other words, in our real-world implementation, the white-box adversarial streaming model does not allow us to use any structured randomness (such as with hidden planted secrets). In {\it the proof}, however, we could introduce these secrets which will be important for our construction.
\item The hash should satisfy collision-resistance.
\end{itemize}

\paragraph{Constructing Streaming Friendly Hashing from FHE.} How could we build such a hash function?  The SIS based hash function does have a fast update time. On the other hand, we could take an arbitrary collision resistant function that compresses the number of bits by a factor of $1/2$ and use it to perform a Merkle tree style hashing to hash strings of arbitrary length. Such a hash will have $\poly(\lambda)$-sized hashing keys. Unfortunately, the first idea suffers from a large hashing key size, which could be shortened in the random oracle (RO) model by using the RO to store the SIS matrix. Moreover, the second proposal will not support positive and negative stream updates to the coordinates of the underlying vector. 

Our main idea is to leverage a suitably chosen Fully Homomorphic Encryption (FHE) scheme satisfying mild structural properties (that are satisfied currently by most schemes based on LWE and Ring-LWE) to build such a hash function. Our construction bears a resemblance to the construction of a collision resistant function from a private information retrieval scheme, where the private information retrieval scheme is implemented by a fully homomorphic encryption scheme.

We directly explain the idea in the context of our construction.
In our hash function, the hash key consists of $L=O(\log n)$ ciphertexts $\ct_1,\ldots,\ct_{L}$. We assume that our FHE scheme has pseudorandom ciphertexts and pseudorandom public keys. In our actual scheme, these ciphertexts will be generated as random strings. It is only in the proof that they will correspond to encryptions of carefully chosen values. 

The idea is that from these $L$ ciphertexts, we will derive $n$ evaluated ciphertexts $\hat{\ct}_j$ for $j\in [n]$ by evaluating $L$ ciphertexts on simple functions $C_1,\ldots,C_n$ (described later). 
 These $n$ ciphertexts will correspond to columns of the SIS type hashing. 
The hash of a vector $\vec{x}$ is computed as the linear combination $\sum_{i \in [n]}\hat{\ct_i}\vec{x}_i$ where $\vec x_i$ is the $i^{th}$coordinate. Here the sum is done over the field on which the FHE ciphertexts live.

The circuits $C_1,\ldots,C_n$ are chosen keeping the security proof in mind. The ciphertexts $\ct_1,\ldots,\ct_{L}$ are ``thought" to encrypt a value $m\in \{0,1\}^{L}$ that will be chosen in the proof, and each $C_i(\cdot)$ on input $m$ outputs $1$ if $i=m$ and $0$ otherwise. Observe that these circuits are polynomial in $L=O(\log n)$ sized and therefore can be evaluated on $\ct_1,\ldots,\ct_{L}$ in polynomial in $(\lambda, \log n)$ time guaranteeing fast updates.

How do we prove security? The idea is that if an adversary breaks collision-resistance it will produce two vectors 
 $\vec{x} \neq \vec{x'}$, so that the corresponding hashes $\sum_{i \in [n]}\hat{\ct_i}\vec{x}_i=\sum_{i \in [n]}\hat{\ct_i}\vec{x}'_i$.
Suppose one could know exactly one coordinate $v\in [n]$ so that $\vec x_v\neq \vec x'_v$. Then we could set $m=v$. In this case $\hat{\ct_v}$ encrypts $1$ and the other ciphertexts encrypt $0$. In this case, 
 $\sum_{i \in [n]}\hat{\ct_i}\vec{x}_i=\sum_{i \in [n]}\hat{\ct_i}\vec{x}'_i$ would lead to a contradiction if our FHE scheme has a special property (that is satisfied by common FHE schemes). The property is that if $\hat{\ct_i}$ encrypts bit $\mu_i\in \{0,1\}$ and $\vec x_i$ are small norm then there is a decryption algorithm that uniquely binds $\sum_{i \in [n]}\hat{\ct_i}\vec{x}_i$ to $\sum_{i\in [n]}\mu_i\vec{x}_i$. If this holds, then the equality $\sum_{i \in [n]}\hat{\ct_i}\vec{x}_i=\sum_{i \in [n]}\hat{\ct_i}\vec{x}'_i$ would imply $\vec x_v=\vec x'_v$ which is a contradiction.

This leads us to two remaining problems:
\begin{enumerate}
    \item Given that we do not know $\vec{x}$ and $\vec{x'}$ ahead of time, how do we decide which index $m$ to encrypt such that it will help identify whether $\vec{x} = \vec{x'}$?

    \item Moreover, a WAR algorithm cannot actually encrypt FHE ciphertexts, since all randomness and parameters used by the algorithm will be exposed to the adversary. Instead, we want our digest to be unstructured and truly random.
\end{enumerate}

\paragraph{FHE with a Pseudorandom Property.}
We solve both issues above by considering a pseudorandom property of the FHE scheme and show that a random guess of the index $m$ actually suffices. We assume that the distributions of public keys and ciphertexts of the FHE scheme are indistinguishable, respectively, from some truly random distributions from the perspective of the adversary. 

We sample truly random digests in our construction, and perform a random guess of $m \in [n]$ in the proof. If an adversary finds a collision pair $\vec{x} \neq \vec{x'}$ with probability $\gamma$, then with probability $p \geq \gamma/n$, our guess pins down a coordinate of disagreement $\vec{x}_m \neq \vec{x'}_m$. In the proof, we then switch the ciphertexts to encrypt $m$. The chance of guessing the disagreement is still $p$, as otherwise one can distinguish the pseudorandom ciphertext. However, as argued earlier, $\vec{x}_m \neq \vec{x'}_m$ roughly implies $\sum_{i \in [n]}\hat{\ct_i}\vec{x}_i \neq \sum_{i \in [n]}\hat{\ct_i}\vec{x'}_i$, and thus two distinct vectors cannot form a hash collision that fools our ``tester''.

In our complexity analysis, we instantiate the FHE scheme using GSW \cite{C:GenSahWat13}, which satisfies the additional properties required by our construction. We remark that in addition to GSW, most current FHE schemes based on LWE and Ring-LWE also exhibit these properties, and thus they are not an artifact obtained from a particular FHE scheme.

\paragraph{Ring-LWE for the Distributed Setting.}
In the distributed model, each remote server receives as input a length-$n$ vector. The goal is to design a protocol such that each server can efficiently compute a short summary of their partition of data. Adapting the above streaming algorithm to this setting, each server needs to evaluate $n$ ciphertexts in order to compute a linear combination $\sum_{i \in [n]}\hat{\ct_i}\vec{x}_i$. 

Using a polynomially secure LWE-based FHE scheme, it is unclear how to speed up such $n$-many individual evaluations as each computation takes $\poly(\lambda)$ operations, resulting in a total complexity of $n\poly(\lambda)$, where $\lambda$ is the security parameter. Therefore, we appeal to the Ring-LWE assumption, which is an algebraic variant of LWE that works over the ring $\mathbb{Z}_{p}[x]/(x^N+1)$. 
Over this ring, we can exploit the SIMD property of Ring-LWE by packing the data as ring elements, i.e., we represent the input vectors as partitioned segments of length proportional to the security parameter (rather than the vector length $n$), and efficiently operate on them segment-by-segment using the Fast Fourier Transform.  This improves the processing time of each server to be almost linear in the dimension of the data. We refer the reader to Section \ref{sec:dis} for more details.

\paragraph{Summary.} Leveraging FHE or lattice-based assumptions in algorithm design are relatively new techniques. We think it would be good to explore their applications more broadly in robust algorithms, beyond the recovery problems or even the adversarial settings considered here. 

%% file: content/preliminaries.tex
\section{Preliminaries}

\paragraph{Basic Notations.} 
For a finite set $S$, we write $x \rsmpl S$ to denote a uniform sample $x$ from $S$. Our logarithms are in base $2$. For $n \in \R$, we use $\ceil{n}$ to denote rounding $n$ up to the nearest integer. For a vector $\vec{x}$, $\lVert \vec{x} \rVert_0$ denotes the $\ell_0$ norm of $\vec{x}$, which is the number of its non-zero entries. We abbreviate PPT for probabilistic polynomial-time. A function $\negl: \N \to \R$ is negligible if for every constant $c > 0$ there exists an integer $N_c$ such that $\negl(x) < x^{-c}$ for all $x > N_c$. Let $\mathds{1}(p)$ for some proposition $p$ be an indicator variable, which equals $1$ if $p$ is true and $0$ otherwise. We use $\vec{b}_i$ to denote the standard basis vector, whose $i$-th entry equals $1$ and $0$ anywhere else.

\subsection{Lattice Preliminaries}

\paragraph{General definitions.}
A \emph{lattice} $\cL$ is a discrete subgroup of $\R^{g}$, or equivalently the set $\cL(\vec{b}_{1},\dots,\vec{b}_{m})=\left\{ \sum_{i=1}^{g}x_{i}\vec{b}_{i} ~:~ x_{i}\in\Z\right\}$ of all integer combinations of $g$ vectors
$\vec{b}_{1},\dots,\vec{b}_{m} \in \R^{g}$ that are linearly independent. Such $\vec{b}_i$'s form a \emph{basis} of $\cL$. The lattice $\cL$ is said to be \emph{full-rank} if $g=m$.
We denote by $\lambda_{1}(\cL)$ the so-called first minimum of $\cL$, defined to be the length of a shortest non-zero vector of $\cL$.

\noindent
\paragraph{Discrete Gaussian and Related Distributions.}
For any $s>0$, define $\rho_s(\vec{x})=\exp(-\pi\|\vec{x}\|^2/s^2)$ for all $\vec{x}\in\R^g$.
We write $\rho$ for $\rho_1$. For a discrete set $S$, we extend $\rho$ to sets by
$\rho_s(S)=\sum_{\vec{x}\in S}\rho_s(\vec{x})$. Given a lattice $\cL$,
the \emph{discrete Gaussian} $\dgaussian_{\cL,s}$ is the distribution over $\cL$
such that the probability of a vector $\vec{y}\in\cL$ is proportional to $\rho_s(\vec{y})$:
  $\Pr_{X\leftarrow \dgaussian_{\cL,s}}[X=\vec{y}]=\frac{\rho_s(\vec{y})}{\rho_s(\cL)}.$

The discrete Gaussian is not bounded but with overwhelming probability a sample from $\dgaussian_{\cL,s}$ is bounded by $s\sqrt{g}$ in $\ell_2$ norm.

\paragraph{Lattice Problems.}

We consider worst-case and average-case problems over lattices. We start by defining the gap version of the shortest vector problem.

\begin{definition}
For any function $\gamma = \gamma(g) \ge 1$, the decision problem $\GapSVP_\gamma$ (Gap Shortest Vector Problem) is defined as follows: the input is a basis $\vec{B}$ for a lattice $\cL \subset \R^g$ and a real number $d > 0$ that satisfy the promise that $\lambda_1(\cL(\vec{B})) \le d$, or $\lambda_1(\cL(\vec{B})) \ge \gamma d$. The goal is to decide whether $\lambda_1(\cL(\vec{B})) \le d$, or $\lambda_1(\cL(\vec{B})) \ge \gamma d$.
\end{definition}

This problem is known to be $\mathsf{NP}$ hard for $\gamma=g^{o(1)}$ \cite{FOCS:Khot04a,STOC:HavReg07}, and is known to be inside $\mathsf{NP}\cap \mathsf{coNP}$ for $\gamma=O(\sqrt{g})$ \cite{Aharonov2005LatticePI,GG00}. The problem is believed to be intractable by polynomial time adversaries even for $\gamma=2^{g^{\epsilon}}$ for $0<\epsilon <1$.

Next, we define the Learning with Errors problem $\mathsf{LWE}$.

\begin{definition} [Learning with Errors]
\label{def:lwe}
Suppose we are given an integer dimension $g\in \mathbb{N}$, a modulus $q\in \mathbb{N}$, a sample complexity $h\geq \Omega( g \cdot \log q)$ and a parameter $\sigma \in \mathbb{R}^{\geq 0}$. Define an error distribution as the discrete Gaussian $\chi^h_\sigma := \dgaussian_{\sigma, \Z^{h}}$. We define $\mathsf{LWE}_{g,h, q,\chi^h_\sigma}$ as follows: sample a matrix $\mat{A}\leftarrow \Z^{g\times h}_{q}$ and a secret $\vec s \in \Z^{1\times g}_{q}$. Sample an error vector $\vec e \gets \chi^h_\sigma$. The problem is to distinguish $(\mat{A}, \vec s \mat{A}+\vec e \mod q)$ from $(\mat{A},\vec r)$ where $\vec r$ is random in $\Z^{1\times h}_{q}$.   
\end{definition}

Under quantum and classical reductions it can be shown that $\mathsf{LWE}_{g,h, q,\chi_{\sigma}}$ is harder to solve than $\mathsf{Gap\mbox{-}SVP}_{\gamma}$ with $\gamma=g\frac{q}{\sigma}$.

We also consider the Ring-LWE problem, which is an average-case lattice problems over ideal lattices. 

\begin{definition} [Ring-LWE]
 Let $q\geq 1$ be a prime number, and let $g$ be a power of $2$. Let $g,\sigma$ be positive integers. The $\mathsf{Ring\mbox{-}LWE}_{g,h,q,\chi_\sigma}$ problem is defined as follows: Consider $R=\frac{\mathbb{Z}_q[x]}{x^{g}+1}$. Sample at random $\vec a_1,\cdots,\vec a_h \leftarrow R$ and a secret polynomial $\vec s \leftarrow R$. Sample $h$ error polynomials $\vec e_1,\cdots, \vec e_h$ such that each $\vec{e}_i$ is a $g$-dimensional vector $\vec e_i\leftarrow \chi_\sigma = \dgaussian_{\sigma,\Z^g}$. The problem is to distinguish the tuple $(\vec a_1,\cdots,\vec a_h, \{\vec s \vec a_i+\vec e_i\}_{i\in [h]})$ from $(\vec a_1,\cdots,\vec a_h, \{\vec r_i\}_{i\in [h]})$ where $\vec r_1,\cdots,\vec r_h$ are random polynomials in $R$.
\end{definition}

The Ring-LWE problem is as hard as worst-case problems over ideal lattices \cite{Ec:LyuPeiReg10}. Ring LWE is sometimes advantageous to use over LWE because of shorter parameters. One ring LWE element is of dimension $g$, and requires only one ring element as a coefficient, as opposed to $g$ vectors in the case of LWE. Moreover, ring multiplication can be sped up using number-theoretic transforms \cite{Eprint:SmaVer11}. 

\subsection{Fully Homomorphic Encryption}\label{sec:fhe}


\paragraph{Notation.} We denote the security parameter of a cryptographic scheme by $\secparam$. We write PPT$_\secparam$ to denote probabilistic poly-time in $\secparam$. We say two ensembles of distributions $\{\calD_\lambda\}_{\lambda \in \N}$ and $\{\calD'_\lambda\}_{\lambda \in \N}$ are computationally indistinguishable, denoted by $\approx_c$, if for any non-uniform PPT adversary $\calA$ there exists a negligible function $\negl$ such that $\calA$ can distinguish between the two distributions with probability at most $\negl(\lambda)$.

\






\begin{remark}\label{rem:indist}
    We also define {\it $(\calT,\epsilon)$-indistinguishability}, where $\calT$ and $\epsilon$ are two functions in $\secparam$. We say that two distribution ensembles $\{\calD_\lambda\}_{\lambda \in \N}$ and $\{\calD'_\lambda\}_{\lambda \in \N}$ are $(\calT, \epsilon)$-indistinguishable, denoted by $\approx_{(\calT, \epsilon)}$,  if there exists an integer $\secparam_0$ such that for all $\secparam > \secparam_0$, for every adversary $\A$, $\calA$ can distinguish between the two distributions with probability at most $\epsilon(\lambda)$.

\end{remark}


\begin{definition}[FHE]\label{def:FHE}
    A fully homomorphic encryption $\fhe$ scheme is a tuple of \emph{PPT$_\secparam$} algorithms $\fhe = (\fhe.\setup, \fhe.\enc, \fhe.\eval, \fhe.\dec)$ with the following properties:

    \begin{itemize}[leftmargin=1em]
        \item $\fhe.\setup(1^{\lambda}) \to (\sk, \pk)$ : On input a security parameter $\lambda$, the setup algorithm outputs a key pair $(\pk, \sk)$.

        \item $\fhe.\enc(\pk, \mu) \to \ct$ : On input a public key $\pk$ and a message $\mu \in\{0, 1\}$, the encryption algorithm outputs a ciphertext $\ct$.

        \item $\fhe.\eval(\pk, C, \ct_1, \cdots , \ct_\ell) \to \hat{\ct}$ : On input a public key $\pk$, a circuit $C : \{0, 1\}^\ell \to \{0, 1\}$ of depth at most $\poly(\lambda)$, and a tuple of ciphertexts $\ct_1, \cdots , \ct_\ell$, the evaluation algorithm outputs an evaluated ciphertext $\hat{ct}$.

        \item $\fhe.\dec(\pk, \sk, \hat{\ct}) \to \hat{\mu}$ : On input a public key $\pk$, a secret key $\sk$, and a ciphertext $\hat{\ct}$ , the decryption algorithm outputs a message $\hat{\mu} \in \{0, 1, \bot\}$.
    \end{itemize}
\end{definition}

We also require an $\fhe$ scheme to satisfy compactness, correctness, and security properties as follows:

\begin{definition}[$\fhe$ Compactness]\label{def:FHEcompactness}
    We say that a $\fhe$ scheme is compact if there exists a polynomial $\poly(\cdot)$ such that for all security parameters $\lambda$, circuits $C: \{0, 1\}^\ell \to \{0, 1\}$ of depth at most $\poly(\secparam)$, and $\mu_i \in \{0, 1\}$ for $i \in [\ell]$, the following holds:
    Given $(\pk, \sk) \gets \fhe.\setup(1^\lambda)$, $\ct_i \gets \fhe.\enc(\pk, \mu_i)$ for $i \in [\ell]$, $\hat{\ct} \gets \fhe.\eval(\pk, C, \ct_1, \cdots, \ct_\ell)$, we always have $\lvert\hat{\ct}\rvert \leq \poly(\lambda)$.
\end{definition}

\begin{definition}[$\fhe$ Correctness]\label{def:FHEcorrectness}
    We say that a $\fhe$ scheme is correct if for all security parameters $\lambda$, circuit $C: \{0, 1\}^\ell \to \{0, 1\}$ of depth at most $\poly(\secparam)$, and $\mu_i \in \{0, 1\}$ for $i \in [\ell]$, the following holds: Given $(\pk, \sk) \gets \fhe.\setup(1^\lambda)$, $\ct_i \gets \fhe.\enc(\pk, \mu_i)$ for $i \in [\ell]$, $\hat{\ct} \gets \fhe.\eval(\pk, C, \ct_1, \cdots, \ct_\ell)$, 
    \[\prob[ \fhe.\dec(\pk, \sk, \hat{\ct}) = f(\mu_1, \cdots, \mu_\ell) ] = 1.\]
\end{definition}

\begin{definition}[$\fhe$ Security]\label{def:FHEsecurity}
    We say that a $\fhe$ scheme satisfies security if for all security parameter $\lambda$, the following holds: For any PPT$_{\secparam}$ adversary $\A$, there exists a negligible function $\negl$ such that
    \[\adv_{\A, \fhe} := 2 \cdot \lvert 1/2 - \Pr[\expt_{\A, \fhe}(1^\lambda) = 1]\rvert < \negl(\lambda),\]
where the experiment $\expt_{\A, \fhe}$ is defined as follows: \vspace{0.2em}

    $\expt_{\A, \fhe}(1^\lambda): $\vspace{-0.7em}
    \begin{enumerate}
        \item On input the security parameter $1^\lambda$, the challenger runs $(\pk, \sk) \gets \fhe.\setup(1^\lambda)$, and $\ct \gets \fhe.\enc(\pk, b)$ for $b \rsmpl \{0, 1\}$. It provides $(\pk, \ct)$ to $\A$.
        
        \item $\A$ outputs a guess $b'$. The experiment outputs $1$ if and only if $b = b'$.
    \end{enumerate}

\end{definition}

\begin{remark}\label{rem:sec}
    The security definition above is in terms of PPT$_{\secparam}$ adversaries and $\negl$ advantage, but it can be lifted to a $(\calT,\epsilon)$ setting similar to Remark \ref{rem:indist} as follows: We say that a $\fhe$ scheme is $(\calT,\epsilon)$-secure if there exists an integer $\secparam_0$, such that for all $\secparam > \secparam_0$, for any non-uniform probabilistic $\calT(\secparam)$ time adversary, the advantage $\adv_{\A, \fhe}$ of the above experiment is at most $\epsilon(\lambda)$.
\end{remark}

\paragraph{Additional properties.} Our construction of robust streaming algorithms will be based on an $\fhe$ scheme satisfying a few additional properties. We refer to this type of $\fhe$ scheme as a {\it pseudorandom $\fhe$ ($\pfhe$) scheme}.

\begin{definition}[Pseudorandom $\fhe$ ($\pfhe$)]\label{def:PFHE}
    Let $\fhe$ be a fully homomorphic encryption scheme \emph{(Definition \ref{def:FHE})}. We call $\fhe$ a $(\calT, \epsilon)$-pseudorandom $\fhe$ scheme if for all security parameters $\lambda$, it satisfies the following additional properties:

    \begin{itemize}[leftmargin=1em]
    \item $\fhe.\eval$ is deterministic.

        \item \emph{(Pseudorandom $\pk$)}. There exists an explicit distribution $\D_{\tilde{\pk}(\lambda)}$, such that the following two distributions are $(\calT, \epsilon)$-indistinguishable:
        
        \[\{\pk : (\pk, \sk) \gets \fhe.\setup(1^\lambda)\} \approx_{(\calT, \epsilon)} \D_{\tilde{\pk}(\lambda)}\]

        \item \emph{(Pseudorandom $\ct$)}. There exists an explicitly distribution $\D_{\tilde{\ct}(\lambda)}$, such that the following two distributions are $(\calT, \epsilon)$-indistinguishable:
        \[\left\{ \ct: \begin{aligned}
            (\pk, \sk) &\gets \fhe.\setup(1^\lambda) \\
            \mu &\in \{0, 1\} \\
            \ct &\gets \fhe.\enc(\pk, \mu) \\
        \end{aligned}\right\} \approx_{(\calT, \epsilon)} \D_{\tilde{\ct}(\lambda)}\]

        \item \emph{($\Z$ Linear Homomorphism)}. Given $(\pk, \sk) \gets \fhe.\setup(1^\lambda)$, the set of all ciphertexts is a subset of a vector space over $\Z_q$ for some $q = q(\secparam)$.
        Let $\calM = \sum_{i \in [k]}x_i\cdot \ct_i \mod q$ be an arbitrary linear combination of ciphertexts. Denote $\mu_i \gets \fhe.\dec(\pk, \sk, \ct_i)$. Suppose $\calM$ satisfies:
        \begin{enumerate}[topsep=0.2em]
            \item $x_i \in \Z$ for all $i \in [k]$
            \item $k \cdot \max_{i \in [k]} x_i \ll q$
        \end{enumerate}
        then we have the following: $\lvert \calM \rvert \leq \poly(\secparam)$. Also, there exists a $\poly(\lambda)$-time decoding algorithm $\ldec$, such that $\Pr[\ldec(\pk, \sk, \calM) = \sum_{i \in [k]} x_i \cdot \mu_i] =1$.
    \end{itemize}

    As in a standard $\fhe$ scheme, we also require $\pfhe$ to satisfy compactness, correctness, and security.
    
\end{definition}

There exist FHE schemes based on LWE and Ring-LWE \cite{C:GenSahWat13,BGV12,BV11F} assumptions. We will focus on the GSW scheme in our streaming construction and the BV scheme in the distributed setting. Both of them satisfy the above definition (or a slight variant) of $\pfhe$.

\subsection{Streaming Algorithms}\label{sec:prelim-stream}

In this section, we recall the streaming model, the structure of a streaming algorithm, and the formal definition of white-box adversarial robustness.

\paragraph{Streaming Model.} The streaming model captures key resource requirements of algorithms for database, machine learning, and network tasks, where the size of the data is significantly larger than the available storage. This model was formalized by Alon, Matias, and Szegedy \cite{10.1145/237814.237823} as a vector undergoing additive {\it updates} to its coordinates. In this model, the input data takes the form of a sequence of updates. These updates assume the existence of a fixed-length underlying vector which represents the current dataset, and dynamically modify this vector one coordinate at a time. 

Given a dimension parameter $n \in \N$, the model behaves as follows:

\begin{itemize}[topsep=0.3em, itemsep=0.3em]
    \item At the beginning of the stream, an underlying vector $\vec{x}$ is initialized to $0^n$.

    \item The input stream is a sequence of $T \leq \poly(n)$ additive updates $(x_1, \cdots, x_T)$. Each update $x_t$ for $t \in [T]$ is a tuple $x_t := (i_t, \Delta_t)$, where $i_t \in [n]$ denotes a coordinate of $\vec{x}$ and $\Delta_t$ is an arbitrary value. 
    
    $x_t := (i_t, \Delta_t)$ is interpreted as updating $\vec{x}_{i_t} \gets \vec{x}_{i_t} + \Delta_t$. Thus at any time $t' \in [T]$ during the stream, $\vec{x}$ is the accumulation of all past updates, with its $i$-th coordinate for $i \in [n]$: $\vec{x}_i = \sum_{t \leq t'} \mathds{1}(i_t = i_{t_0})\cdot \Delta_t.$
\end{itemize}

A streaming algorithm receives the stream of updates and uses limited memory to compute a function of $\vec{x}$. We stress that during the stream, the underlying vector $\vec{x}$ is not explicitly available to the algorithm, but is instead implicitly defined by the sequence of past updates. 

In this work, we consider a bounded integer stream, which assumes that throughout the stream, $\vec{x}$ is promised to be in the range $\{-N, -N+1,\cdots, N-1, N\}^n$ for some integer $N \in \poly(n)$. 

\paragraph{Streaming Algorithm.}

 A streaming algorithm maintains some internal state during the stream, which evolves upon receiving new updates. It takes one pass over all updates and then uses its maintained state to compute a function of the underlying dataset $\vec{x}$. In the streaming model, the dimension parameter $n$ (the length of $\vec{x}$) is typically large. Therefore, a streaming algorithm usually optimizes the size of the state so that it takes a sublinear number of bits of space in $n$ to represent.

 We define the general framework for streaming algorithms as follows:

\begin{definition}[Streaming Algorithm]\label{def:SA}
    Given a query function ensemble $\mathcal{F} = \{f_n : \Z^n \to X\}_{n \in \N}$ for some domain $X$, a streaming algorithm $\sa$ that computes $\mathcal{F}$ contains a tuple of operations $\sa$ $=$ $(\sa.\setup$, $\sa.\update, \sa.\report)$, with the following properties:

    \begin{itemize}
        \item $\sa.\setup(n)$ : On input a dimension parameter $n$, the setup operation initializes an internal data structure $\ds$ within the memory of $\sa$.

        \item $\sa.\update(x_t)$ : On input an additive update $x_t := (i_t, \Delta_t)$, the update operation modifies $\ds$.

        \item $\sa.\report() \to r$: The report operation uses $\ds$ to compute a query response $r \in X$. 
        
        A response $r$ is said to be correct if $r = f_n(\vec{x})$, where $\vec{x} \in \Z^n$ is the accumulation of all past updates such that $\vec{x}_i = \sum_t \mathds{1}(i_t = i)\cdot \Delta_t$ for $i \in [n]$.
    \end{itemize}
\end{definition}


In this work, we only consider exact computations of query functions ($r = f_n(\vec{x})$), though the above definition can also be generalized to approximation problems ($r \approx f_n(\vec{x})$).

\paragraph{White-box Adversarial Stream.}
In the conventional {\it oblivious} stream model, the sequence of updates may be chosen adversarially, but is assumed to be chosen independently of the randomness of a streaming algorithm $\sa$. In contrast, in the white-box adversarial streaming model, a sequence of stream updates $(x_1, \cdots, x_T)$ is adaptively generated by an adversary who sees the {\it full internal state} of $\sa$ at all times, including $\ds$ and any other parameters and randomness used by $\sa$. We say that a streaming algorithm is {\it white-box adversarially robust} if it guarantees to output correct responses even against such powerful adversaries. This is described as the following game:

\begin{definition}[White-Box Adversarially Robustness]\label{def:wbRobust}
    We say a streaming algorithm $\sa$ is white-box adversarially robust (WAR) if for all dimension parameters $n \in \N$, the following holds: For any PPT adversary $\A$,
    the following experiment $\expt_{\A, \sa}(n)$ outputs $1$ with probability at most $\negl(n)$: 
    
    $ $\vspace{-0.5em}
    
    $\expt_{\A, \sa}(n): $
    \begin{enumerate}
        \item The challenger and $\A$ agree on a query function ensemble $\mathcal{F} = \{f_n : \Z^n \to X\}_{n \in \N}$ for some domain $X$. On input a dimension parameter $n$, let $f_n \in \mathcal{F}$ be the query function of the experiment.

        \item The challenger runs $\sa.\setup(n)$. It then provides all internal states of $\sa$ to $\A$, including $\ds$ and any other parameters and randomness previously used by $\sa$.

        \item For some $T \in \poly(n)$, at each time $t \in [T]$: 
        \begin{itemize}[label=-, leftmargin=*]
            \item $\A$ outputs an adaptive update $x_t := (i_t, \Delta_t)$ for $t \in [T]$. Also, $\A$ may issue a query. 
        
            \item The challenger runs $\sa.\update(x_t)$, and outputs a response $r_t \gets \sa.\report()$ if it is queried. 
            
            The challenger then provides all internal states of $\sa$ to $\A$, including $\ds$ and any other parameters and randomness used by $\sa$.
        \end{itemize}

        \item The experiment outputs $1$ if and only if at some time $t\in [T]$ in the stream, $\sa$ output an incorrect response $r_t \neq f_n(\vec{x})$, where $\vec{x}$ is the underlying data vector with $\vec{x}_i = \sum_{j \in [t]} \mathds{1}(i_j = i)\cdot \Delta_j$ for $i \in [n]$.
    \end{enumerate}
\end{definition}

%% file: content/streaming.tex
\section{Streaming $K$-Sparse Recovery}

\paragraph{Sparse Recovery Problem.}

We study the sparse recovery problem for vectors in white-box adversarial streaming settings. There exist deterministic (thus WAR) streaming algorithms that recover $k$-sparse vectors {\it assuming that the input vector is $k$-sparse }. However, their outputs can be arbitrary when the input violates such an assumption. In contrast, our algorithm has the crucial property that it can detect if the input violates the sparsity constraint. We describe the problem as follows:

\begin{definition}[$k$-Sparse Recovery]\label{def:rec1}
    Given a sparsity parameter $k$ and an input vector $\vec{x} \in \mathbb{Z}^n$, output $
    \vec{x}$ if $\lVert \vec{x}\rVert_0 \leq k$ and report $\bot$ otherwise.
\end{definition}

Here and throughout, for the sparsity to make sense we always assume $k \leq n$, though $k$ can be a function in $n$.
Before tackling this problem, we first consider a relaxation for it, which only guarantees correct recovery when the input vector $\vec{x}$ is $k$-sparse. In other words, we have no constraint on the function output when $\lVert \vec{x}\rVert_0 > k$. 

\begin{definition}[Relaxed $k$-Sparse Recovery]\label{def:rrec1}
    Given a sparsity parameter $k$ and an input vector $\vec{x} \in \mathbb{Z}^n$, output $\vec{x} \in \Z^n$ if $\lVert \vec{x}\rVert_0 \leq k$.
\end{definition}

As mentioned earlier, there exists a space- and time-efficient, deterministic streaming algorithm that solves the relaxed $k$-sparse recovery problem, which we will refer to as $\sa_0$ in our algorithm. With $\sa_0$ and a $\pfhe$, we are now ready to construct a WAR streaming algorithm for the $k$-sparse recovery problem.

\begin{const}\label{const1}

    Let $\sa_0$ be an algorithm for the Relaxed $k$-Sparse Recovery problem. Let $\pfhe$ be a pseudorandom $\fhe$. Let $\lambda = \calS(n)$ be the security parameter of $\pfhe$, for some function $\calS$ in $n$ that we will specify later. Choose $q \in \poly(n)$ with $q \gg n\cdot N$ to be the modulus of $\pfhe$ ciphertexts, where $N \in \poly(n)$ is a bound on the stream length. All computations will be modulo $q$.
    Our construction is as follows:
    \begin{itemize}[leftmargin=1em]
        \item $\sa.\setup(n)$:
        
        Run $\sa_0.\setup(n)$.
        
        Define $\ell = \ceil{\log n}$. 
        For each $i \in [n]$, define a circuit $C_{i}: \{0, 1\}^\ell \to \{0, 1\}$, such that:      
        \[C_{i}(\mu_1, \mu_2, \cdots, \mu_\ell) = \begin{cases}
            1 & \text{if $(\mu_1,\mu_2, \cdots, \mu_\ell)$ is the bit representation of $i$} \\
            0 & \text{otherwise}.
        \end{cases}\]
        
        Take $\D_{\tilde{\pk}(\lambda)}$ to be the distribution of pseudo-public key in Definition \ref{def:PFHE}. Sample $\tilde{\pk} \gets \D_{\tilde{\pk}(\lambda)}$.
        Similarly, take $\D_{\tilde{\ct}(\lambda)}$ to be the distribution of pseudo-ciphertext in Definition \ref{def:PFHE}. Sample $\tilde{\ct_1}, \tilde{\ct_2}, \cdots, \tilde{\ct_\ell} \gets \D_{\tilde{\ct}(\lambda)}$. Store $(\tilde{\pk}, \tilde{\ct_1}, \cdots, \tilde{\ct_\ell})$ and $sketch := 0$.
        
        \item $\sa.\update(i_t, \Delta_t)$: 
        
        Run $\sa_0.\update(i_t, \Delta_t)$. 
        
        Evaluate $\hat{\ct_{i_t}} \gets \pfhe.\eval(\tilde{\pk}, C_{i_t}, \tilde{\ct_1}, \tilde{\ct_2}, \cdots, \tilde{\ct_\ell})$. Then update \[sketch \gets sketch + \Delta_t \cdot \hat{\ct_{i_t}}.\]

        \item $\sa.\report()$: 
        
        Let $\vec{x}' \gets \sa_0.\report()$. 
        
        If $\lVert \vec{x}'\rVert_0 > k$, output $\bot$. Otherwise, evaluate $\hat{\ct_i} \gets \pfhe.\eval(\tilde{\pk}, C_{i}, \tilde{\ct_1}, \tilde{\ct_2}, \cdots, \tilde{\ct_\ell})$ for all $i \in [n]$. Then compute \[hash \gets \sum_{i \in [n]} \hat{\ct_i}\cdot \vec{x'}_i.\]
        Output $\vec{x}'$ if $hash = sketch$, and output $\bot$ otherwise.
        \end{itemize}
    \end{const}

\paragraph{Robustness.}

Recall that in the $\fhe$ Preliminary section, we define the $(\calT, \epsilon)$-security of a $\fhe$ scheme, where $\calT$ and $\epsilon$ are functions in $\secparam$. Also, in Construction \ref{const1}, we set $\secparam = \calS(n)$ for some function $\mathcal{S}$. We will show that given appropriate relationships between these functions, the $\sa$ scheme in Construction \ref{const1} is WAR. We state our main theorem as follows:

\begin{theorem}[WAR of Construction \ref{const1}]\label{thm:stream}
    Suppose $\sa_0$ is a deterministic, polynomial time streaming algorithm for Relaxed $k$-Sparse Recovery, and $\pfhe$ is a pseudorandom $\fhe$ scheme with $(\calT, \epsilon)$-security. Let  $\secparam = \calS(n)$ be the $\pfhe$ security parameter for some function $\mathcal{S}$. 
    
    If for all $n \in \N$, functions $\calT, \epsilon, $ and $\calS$ satisfy: 
    
    \begin{itemize}[topsep=0.2em, itemsep=0.1em, label=\textperiodcentered]
        \item $\calT(\mathcal{S}(n)) = \Omega(\poly n)$

        \item $\epsilon(\mathcal{S}(n)) \leq \negl(n)$ for some negligible function $\negl$.
    \end{itemize}
    
    then the $\sa$ scheme in \emph{Construction \ref{const1}} is a WAR streaming algorithm for the $k$-Sparse Recovery problem.
  
\end{theorem}

    \begin{remark}
    Roughly, we need $\poly(n) \subseteq \calT(\mathcal{S}(n))$ to convert between the runtime of a $\pfhe$ adversary and a $\sa$ adversary. We also need $\epsilon(\mathcal{S}(n)) \leq \negl(n)$ to convert between the advantage of a $\pfhe$ adversary and a $\sa$ adversary.
    \end{remark}

    \begin{proof}[Proof (of Theorem \ref{thm:stream})]

We first show that at any time during the stream, if the underlying vector $\vec{x}$ is $k$-sparse, then the query response of $\sa$ is correct. 

\begin{lemma}\label{lem:sparseCase}
    Suppose $\pfhe$ satisfies \emph{Definition \ref{def:PFHE}} and $\sa_0$ solves the Relaxed $k$-Sparse Recovery problem. Then the following holds: at any time $t \in [T]$, if the current underlying vector $\vec{x}$ satisfies $\lVert \vec{x}\rVert_0 \leq k$, then $\sa.\report$ correctly outputs $\vec{x}$.
\end{lemma}

\begin{proof}
    By the correctness of the $\sa_0$ scheme, when $\lVert \vec{x}\rVert_0 \leq k$, $\sa_0.$ $\report$ outputs $\vec{x}'=\vec{x}$. Since Definition \ref{def:PFHE} restricts $\pfhe.\eval$ to be deterministic, we essentially have $hash = \sum_{i \in [n]} \hat{ct}_i \cdot \vec{x'}_i \hspace{0.3em} \text{ and } \hspace{0.3em} sketch = \sum_{i \in [n]} \hat{ct}_i \cdot \vec{x}_i$ for the same set of ciphertexts $\hat{ct}_i$.
Thus $hash = sketch$ follows from $\vec{x}'=\vec{x}$. In this case, $\sa.\report$ correctly outputs $\vec{x}$ as desired.

\end{proof}

It remains to show the correctness of responses in the case when the underlying data vector is denser than $k$. Our proof proceeds via a sequence of hybrid experiments between an adversary $\A$ and a challenger.

\begin{itemize}[leftmargin=*]
    \item $\hybrid_0$: This is a real WAR experiment $\expt_{\A, \sa}(n)$ from Definition \ref{def:wbRobust}, where the challenger runs the $\sa$ scheme from Construction \ref{const1}. On input a dimension parameter $n$, the challenger samples $\tilde{\pk} \gets \D_{\tilde{\pk}(\lambda)}$ and $\tilde{\ct_1}, \tilde{\ct_2}, \cdots, \tilde{\ct_{\ceil{\log n}}}  \gets \D_{\tilde{\ct}(\lambda)}$ from the distributions of pseudo-public key and pseudo-ciphertext (Definition \ref{def:PFHE}) and provides them to $\A$. 
    
    $\A$ outputs $T \in \poly(n)$ -many adaptive updates $x_t:=(i_t, \Delta_t)$ upon seeing past internal states of $\sa$. At any time $t \in [T]$, these updates implicitly define an underlying vector $\vec{x}$, whose $i$-th coordinate $\vec{x}_i = \sum_t \mathds{1}(i_t = i)\cdot \Delta_t$ for $i \in [n]$.

    \item $\hybrid_1$: This is the same experiment as $\hybrid_0$, except that $\A$ now runs some additional extraction on itself: 
    \begin{itemize}[label=-]
        \item At each time $t \in [T]$, $\A$ computes the underlying vector $\vec{x}$ defined by all of its past updates. 

        \item $\A$ runs an additional instance of the $\sa_0$ scheme from Definition \ref{def:rrec1} (which solves the Relaxed k-Sparse Recovery problem): 
        
        \begin{enumerate}[itemsep=0.4em]
            \item Before $\A$ generates its first update, it initializes its $\sa_0$ instance to be in the same state as the $\sa_0$ scheme run by the challenger. (i.e., $\A$ runs $\sa_0.\setup(n)$ using the same set of randomness used by the challenger.)
            
            \item At each time $t \in [T]$, after $\A$ generates an update $x_t$, it calls \newline $\sa_0.\update(x_t)$ and then $\sa_0.\report()$ to get a query response $\vec{x}'$.
        \end{enumerate}
    \end{itemize}
\end{itemize}

We note that if $\A$ runs in $\poly(n)$ time in $\hybrid_0$, then it still runs in $\poly(n)$ time in $\hybrid_1$ since computing $\vec{x}$ is trivial and $\sa_0$ runs in poly-time.

For an adversary $\A$, we write $\hyb_i(\A)$ to denote the output of $\hybrid_i$.

\begin{lemma}\label{lem:hyb01}
    For any adversaries $\A$, $\lvert \Pr[\hyb_0(\A) = 1] - \Pr[\hyb_1(\A) = 1] \rvert = 0$.
\end{lemma}

\begin{proof}
    This just follows from the fact that the additional extractions run by $\A$ do not affect the output of the experiments.
\end{proof}

\begin{itemize}[leftmargin=*]
    \item $\hybrid_2$: Same as $\hybrid_1$, except that at the beginning of the experiment, the challenger samples a random coordinate $m \rsmpl [n]$ which will be hidden from $\A$. The experiment outputs $1$ if at some time $t \in [T]$ both of the following holds:

    \begin{enumerate}
        \item $\sa$ outputs an incorrect response $r_t \neq f_n(\vec{x})$ (i.e., the same condition for $\hyb_1(\A) = 1$), 

        \item The $m$-th coordinates of $\vec{x}$ and $\vec{x'}$ extracted by $\A$ satisfy $\vec{x}_m \neq \vec{x'}_m$,
        
    \end{enumerate}

    \end{itemize}

We show that with the additional success condition, the success probability of $\hybrid_2$ drops by at most a $1/n$ multiplicative factor compared to $\hybrid_1$.

\begin{lemma}\label{lem:hyb12}
    For any adversaries $\A$, $\Pr[\hyb_2(\A) = 1] \geq \Pr[\hyb_1(\A) = 1]/n$.
\end{lemma}

Before proving Lemma \ref{lem:hyb12}, we make the following observation:

\begin{claim}\label{clm:hybClaim}
    At any time $t \in [T]$ during the $\hybrid_2$ experiment, if $\sa$ outputs an incorrect response $r_t \neq f_n(\vec{x})$, and $\vec{x}$ and $\vec{x'}$ are the vectors extracted by $\A$ at time $t$, then $\vec{x'} \neq \vec{x}$.
\end{claim}

\begin{proof}
    By Lemma \ref{lem:sparseCase}, $r_t \neq f_n(\vec{x})$ can occur only when $\lVert \vec{x}\rVert_0 > k$ and $r_t \neq \bot$. In this case, by the construction of $\sa.\report$, $r_t$ should be a $k$-sparse vector, and thus different from $\vec{x}$. Also, using deterministic $\sa_0.\update$ and $\sa_0.\report$, $\A$ exactly extracts $\vec{x'} = r_t$, so $\vec{x'} \neq \vec{x}$.
\end{proof}

Lemma \ref{lem:hyb12} then follows from Claim \ref{clm:hybClaim}:

\begin{proof}[Proof (of Lemma \ref{lem:hyb12}).]
     Claim \ref{clm:hybClaim} states that $\vec{x'} \neq \vec{x}$, and in particular, $\vec{x'}$ and $\vec{x}$ disagree on at least one coordinate. Recall that the random coordinate $m$ is hidden from $\A$ at all times, so with probability at least $1/n$, $m$ collides with a coordinate where $\vec{x'}$ and $\vec{x}$ disagree. In this case, both conditions in $\hybrid_2$ are satisfied, so $\Pr[\hyb_2(\A) = 1 \mid \hyb_1(\A) = 1] \geq 1/n$.
\end{proof}

\begin{itemize}[leftmargin=*]
    \item $\hybrid_3$: Same as $\hybrid_2$, except that the challenger now generates an actual public key $\pk: (\pk, \sk) \gets \pfhe.\setup(1^\lambda)$ to replace $\tilde{\pk}$.

    \item $\hybrid_4$: Same as $\hybrid_3$, except that the challenger now uses the random coordinate $m \rsmpl [n]$ (whose bits are denoted as $m_1 m_2\cdots m_{\ceil{\log n}}$) that was sampled privately as the message in the ciphertexts: It generates $\ceil{\log n}$ actual ciphertexts $\ct_i \gets \pfhe.\enc(\pk, m_i)$ for $i \in [\ceil{\log n}]$ to replace $(\tilde{\ct_1}, \tilde{\ct_2}, \cdots, $ $\tilde{\ct_{\ceil{\log n}}})$. i.e., $(\ct_1, \ct_2, \cdots, \ct_{\ceil{\log n}})$ encrypts the bit representation of $m$.
    
\end{itemize}

\begin{lemma}\label{lem:hyb23}
    Given that $\pfhe$ satisfies \emph{Definition \ref{def:PFHE}}, for any PPT$_n$ adversaries $\A$, $\lvert \Pr[\hyb_2(\A) = 1] - \Pr[\hyb_3(\A) = 1] \rvert \leq \negl(n)$ for some negiligble function $\negl$.
\end{lemma}

\begin{proof}

   The only difference in $\A$'s views in $\hybrid_2$ and $\hybrid_3$ is the way the challenger generates $\tilde{\pk}$ or $\pk$. In $\hybrid_2$, the challenger samples $\tilde{\pk} \gets \mathcal{D}_{\tilde{\pk}(\lambda)}$; and in $\hybrid_3$, the challenger runs $\pk: (\pk, \sk) \gets \pfhe.\setup(1^\lambda)$. Thus $\A$ distinguishing these two hybrids will distinguish the two distributions $\{\pk : (\pk, \sk) \gets \fhe.\setup(1^\lambda)\}$ and $\D_{\tilde{\pk}(\lambda)}$ with the same advantage. Also by the $(\calT, \epsilon)$-indistinguishability in Definition \ref{def:PFHE}, for any $\poly(n) \leq \calT(\lambda)$-sized adversaries, the latter two distributions can be distinguished with probability at most $\epsilon(\lambda) \leq \negl(n)$.
\end{proof}

\begin{lemma}\label{lem:hyb34}
    Given that $\pfhe$ satisfies \emph{Definition \ref{def:PFHE}}, for any PPT$_n$ adversaries $\A$, $\lvert \Pr[\hyb_3(\A) = 1] - \Pr[\hyb_4(\A) = 1] \rvert \leq \negl(n)$ for some negiligble function $\negl$.
\end{lemma}

\begin{proof}
       Similar to Lemma \ref{lem:hyb23}, the only difference in $\A$'s views in $\hybrid_3$ and $\hybrid_4$ is the way the challenger generates $(\tilde{\ct_1}, \cdots, \tilde{\ct_{\ceil{\log n}}})$ or $(\ct_1, \cdots, \ct_{\ceil{\log n}})$. Recall that we have: $\left\{ \ct\right\} \approx_{(\calT, \epsilon)} \D_{\tilde{\ct}(\lambda)}$ from Definition \ref{def:PFHE}. Therefore for any PPT$_n$ adversary $\A$, the advantage of distinguishing $(\tilde{\ct_1}, \tilde{\ct_2}, \cdots, $ $ \tilde{\ct_{\ceil{\log n}}})$ in $\hybrid_3$ from a ciphertext sequence $(\ct_1, \ct_2, \cdots, \ct_{\ceil{\log n}})$ where each $\ct_i \gets \pfhe.\enc(\pk, m_i)$ is at most $\ceil{\log n} \cdot \epsilon(\lambda) = \negl(n)$. Thus $\A$ cannot distinguish between $\hybrid_3$ and $\hybrid_4$ with non-negligible advantage, as otherwise either the pseudorandom-$\ct$ property in Definition \ref{def:PFHE} would be broken.
       
\end{proof}

Now that we have related the success probability between the above consecutive hybrid experiments, we will show next that $\Pr[\hyb_4(\A) = 1] = 0$.

\begin{lemma}\label{lem:hyb4}
    Given that $\pfhe$ satisfies \emph{Definition \ref{def:PFHE}}, for any PPT$_n$ adversaries $\A$, $\Pr[\hyb_4(\A) = 1] = 0$.
\end{lemma}

\begin{proof}
As in the proof for Claim \ref{clm:hybClaim}, if $\hyb_4(\A) = 1$, at some time $t$ we must have $r_t = \vec{x'} \neq \vec{x}$. Additionally, we have:

\begin{itemize}[label=\textperiodcentered, topsep=0.4em, itemsep=0.2em]
    \item $\vec{x'}_m \neq \vec{x}_m$, which is among the condition of $\hyb_4(\A) = 1$, and 
    \item $hash = sketch$, which is verified in $\sa.\report$.
\end{itemize}

Recall the $\sa$ scheme computes $hash := \sum_{i \in [n]} \hat{\ct_i} \cdot \vec{x'}_i$ and $sketch := \sum_{i \in [n]} \hat{\ct_i} \cdot \vec{x}_i$, where $\hat{\ct_i} \gets \pfhe.\eval(\pk, C_{i}, \ct_1, \cdots, \ct_{\ceil{\log n}}), \hspace{0.3em} \text{for } \hspace{0.3em}C_{i}(\mu) = 1$ iff $\mu = i$. Therefore, given that $(\ct_1, \cdots, \ct_{\ceil{\log n}})$ encrypts $m$ in $\hybrid_4$, $\hat{\ct_m}$ should encode $1$, and $\hat{\ct_i}$ encodes $0$ for all $i \neq m$. Thus by the linear homomorphism property of $\pfhe$, there exists a deterministic decoding algorithm such that $hash$ decodes to $\vec{x}'_m$ and $sketch$ decodes to $\vec{x}_m$. However, since $\vec{x'}_m \neq \vec{x}_m$, this essentially implies $hash \neq sketch$, contradicting the condition in $\sa.\report$.

\end{proof}
Finally, summarizing Lemma $\ref{lem:hyb01}$ to Lemma $\ref{lem:hyb4}$:

\begin{enumerate}[leftmargin=*, topsep=0.3em, itemsep=0.2em]
    \item $\lvert \Pr[\hyb_0(\A) = 1] - \Pr[\hyb_1(\A) = 1] \rvert = 0$

    \item $\Pr[\hyb_2(\A) = 1] \geq \Pr[\hyb_1(\A) = 1]/n$

    \item $\lvert \Pr[\hyb_2(\A) = 1] - \Pr[\hyb_3(\A) = 1] \rvert = \negl(n)$

    \item $\lvert \Pr[\hyb_3(\A) = 1] - \Pr[\hyb_4(\A) = 1] \rvert = \negl(n)$

    \item $\Pr[\hyb_4(\A) = 1] = 0$
\end{enumerate}

We have that for any adversaries $\A$, $\Pr[\hyb_0(\A) = 1] \leq \negl(n) \cdot n = \negl(n).$ i.e., The real WAR game between $\A$ and a challenger running the $\sa$ scheme from Construction \ref{const1} outputs $1$ with probability at most $\negl(n)$.

\end{proof}

%% file: content/efficiency.tex
\section{Efficiency}
In this section, we discuss the space and time complexity of Construction \ref{const1}. We start by instantiating $\sa_0$ and $\pfhe$ used in our construction.

\subsection{Deterministic Relaxed Sparse Recovery}\label{sec:instant-sa0}
There exist space- and time-efficient streaming algorithms for the relaxed $k$-sparse recovery problems. Given a sparse vector $\vec{x} \in \Z^n$, many of these algorithms perform some linear measurement $\Phi: \Z^n \to \Z^r$ on $\vec{x}$ with $r \ll n$ to get a compressed representation $\Phi(\vec{x})$, which is often referred to as a {\it sketch}
of $\vec{x}$. This sketch is later used to reconstruct $\vec{x}$. The linearity property is useful for incrementally maintaining the sketch during the stream. 
Many algorithms implement the above specification using a random measurement matrix $\Phi$. However, in our attempt to remove the random oracle assumption within sub-linear total space, we cannot afford to store an $n$-column random matrix during the stream. Instead, we use algorithms which are deterministic, such that the measurements themselves incur no space overhead.

\begin{theorem}[\cite{10.1145/3039872}]\label{thm:rrec1}
    There exists a streaming algorithm $\sa_0$ for the Relaxed $k$-Sparse Recovery problem, such that given an input parameter $n$, for an integer stream with entries bounded by $\poly(n)$:
    \begin{itemize}[topsep=0.2em, itemsep=0.1em, label=\textperiodcentered]
        \item $\sa_0.\setup, \update$, and $\report$ are all deterministic.

        \item $\sa_0$ takes $\OO(k)$ bits of space.

        \item $\sa_0.\update$ runs in $\OO(1)$.

        \item $\sa_0.\report$ runs in $\OO(k^{1+c})$ for an arbitrarily small constant $c > 0$.
    \end{itemize}
    
\end{theorem}

The deterministic Relaxed k-Sparse Recovery algorithm in Theorem \ref{thm:rrec1} achieves nearly optimal space and time. However, the measurement is ``constructed"
non-explicitly using the probabilistic method. For explicit construction, we consider the following result:

\begin{theorem}[\cite{DeterministicCompressedSensing}]\label{thm:rrec2}
    There exists a streaming algorithm $\sa_0$ for $k$-Sparse Recovery without Detection, such that given an input parameter $n$, for an integer stream with entries bounded by $\poly(n)$:

    \begin{itemize}[topsep=0.2em, itemsep=0.1em, label=\textperiodcentered]
       \item $\sa_0.\setup, \update$, and $\report$ are all deterministic.

        \item $\sa_0$ takes $\OO(k)$ bits of space.

        \item $\sa_0.\update$ runs in (amortized) $\OO(1)$.

        \item $\sa_0.\report$ runs in $\OO(k^2)$.
    \end{itemize}
    
\end{theorem}

The above theorem except the amortized runtime has been shown in \cite{DeterministicCompressedSensing}.

\begin{proof}[Proof (of the runtime in \normalfont{Theorem \ref{thm:rrec2}})]
The measurement $\Phi$ in Theorem $\ref{thm:rrec2}$ is a $2k \times n$ Vandermonde matrix explicitly defined \footnote{For our discussion, we choose a Vandermonde matrix that is easy to describe. Though many other Vandermonde matrices also satisfy the specifications.} as $\Phi_{i, j} = j^{i-1}$ over $\Z_p$, where $p > 2M$ is a large enough prime for $M \in \poly(n)$ being the bound on the magnitude of the stream entries. The reconstruction algorithm uses the same idea as the algebraic algorithm for decoding Reed-Solomon codes. It uses the sketch to construct an error-locator polynomial; the roots of this polynomial identify the signals appearing in the sparse superposition.


We note that the Vandermonde matrix is dense but structured. Thus to achieve the claimed $\OO(1)$ update time, we buffer every $2k$ stream updates and then efficiently process them in one batch. We need the following result for fast multi-point evaluation of polynomials and the Transposition Principle:

\begin{theorem}[\cite{gathen_gerhard_2013}]
\label{thm:fast-eval}
    Let $R$ be a ring, and let $q \in R[x]$ be a degree-$d$ polynomial. Then given at most $d$ distinct points $x_1, \cdots, x_d \in R$, all values $q(x_1), \cdots, q(x_d)$ can be computed using $\mathcal{O}(d\log^2 d \log\log d)$ total operations over $R$.
\end{theorem}

\begin{theorem}[Transposition Principle, \cite{SHOUP1994371}]
\label{thm:transpos}
    Let $K$ be a field and $\mat{A}$ be an $r \times w$ matrix over $K$. Suppose there exists an arithmetic circuit of size $s$ that can compute $\mat{A}\vec{y}$ for arbitrary length-$w$ vector $\vec{y}$ over $K$. Then there exists an $\mathcal{O}(s)$-time algorithm that transforms this circuit to compute $\mat{A^\top}\vec{x}$ for arbitrary length-$r$ vector $\vec{x}$ over $K$, with the size of the new circuit in $\calO(s)$.
\end{theorem}

Theorem $\ref{thm:fast-eval}$ and Theorem $\ref{thm:transpos}$ together show the following update time for Vandermonde measurements: 

\begin{proposition}
\label{prop:Vandermonde}
    Suppose $\Phi$ is a $2k \times n$ Vandermonde matrix as described above. Let $x_1, x_2, \cdots, x_{2k}$ be a sequence of updates with $x_t:=(i_t, \Delta_t)$ for each $t \in [2k]$. Define $\vec{v} = \sum_{t \in [2k]} \vec{b}_{i_t} \cdot \Delta_t$, then $\Phi(\vec{v})$ can be computed using $\OO(k)$ time.
\end{proposition}

\begin{proof}[Proof (of \normalfont{Proposition \ref{prop:Vandermonde}})]
     Without loss of generality, we assume all updated coordinates $i_1, i_2, \cdots,$ $i_{2k}$ are distinct, as one can always compress multiple updates to the same coordinate as a single update. Then $i_1, i_2, \cdots, i_{2k}$ index $2k$ columns of $\Phi$ and $2k$ coordinates of $\vec{v}$ (which are the only coordinates that can possibly be non-zero). Let $\mat{A}$ be the indexed $2k \times 2k$ column sub-matrix and  $\vec{v'}$ be the $2k$-length vector containing only $i_1, i_2, \cdots, i_{2k}$-th entries of $\vec{v}$ in their original order. We have $\Phi\vec{v} = \mat{A}\vec{v'}$.

    Observe that the column submatrix $\mat{A}$ is still a Vandermonde matrix, and each row $i$ of $\mat{A}^\top$ takes the form $[1, \alpha_i, \alpha_i^2, \cdots, \alpha_i^{2k-1}]$ for some $\alpha_i \in [n]$. For an arbitrary vector $u \in \Z_p^{2k}$, we interpret it as the coefficient vector of a degree-$2k$ polynomial in $\Z_p^{2k}$. Then $\mat{A}^\top\vec{u}$ can be computed as evaluating the polynomial at $2k$ points $\alpha_1, \alpha_2, \cdots \alpha_{2k}$, which takes $\OO(k)$ time using the fast multi-point evaluation in Theorem \ref{thm:rrec1}. By the Transposition Principle, this implies asymptotically the same runtime for evaluating $\mat{A}\vec{v'}$ thus $\Phi\vec{v}$.
\end{proof}

Therefore, once we buffer the updates into $2k$-sized batches and amortize the cost of processing each batch to the future $2k$ updates, we get $\OO(1)$ amortized runtime for $\sa_0.\update$.

\end{proof}

\subsection{Pseudorandom FHE}\label{sec:instant-pfhe}
    We focus on GSW \cite{C:GenSahWat13}, a standard LWE based fully homomorphic encryption scheme, for our instantiation. 
    To start with, we briefly review the construction of GSW. We omit $\dec$ since it is not involved in our scheme. Instead, we will describe a decoding algorithm $\ldec$ for linear combinations of ciphertexts (which largely resembles the original $\dec$ of GSW) when showing the $\Z$ Linear Homomorphism property of GSW.

    \begin{itemize}[leftmargin=*]
        \item $\setup(1^\secparam)$: Choose a modulus $q$, a lattice dimension parameter $g$, and error distribution $\chi$ parametrized by $\secparam$. Also, choose some $h = \Theta(g \log q)$. 
        Let $params = (g, h, q, \chi)$ be parameters for $\LWE_{g, h, q, \chi}$ .

        \item $\skgen(params)$: Sample $\overline{\vec{s}} \gets \Z_q^{ g}$. Output $\sk = \vec{s}^\top = [-\overline{\vec{s}}^\top \mid 1]$.

        \item $\pkgen(params, \sk)$: Sample a matrix $\overline{\mat{A}} \gets \Z^{g \times h}_q$ uniformly and an error vector $\vec{e} \gets \chi^h$. Let $\vec{\alpha} = \overline{\vec{s}}^\top \cdot\overline{\mat{A}} + \vec{e}^\top$. Set $\pk = \mat{A}$ to be the $(g+ 1)$-row matrix consisting of $g$ rows of $\overline{\mat{A}}$ followed by $\vec{\alpha}$.  

        \item $\enc(params, \pk, \mu)$: Sample a uniform matrix $\mat{R} \gets \{0, 1\}^{h\times h}$. For $\pk = \mat{A}$, output $\mat{A} \cdot \mat{R} + \mu \cdot \mat{G}$ where $\mat{G}$ is a gadget matrix.

    \end{itemize}

    We assume that our set of LWE parameters $params$ is $(\calT, \epsilon)$-secure, such that for any $\calT(\secparam)$ time adversary, the advantage of distinguishing an LWE sample from a uniformly random sample is at most $\epsilon(\lambda)$. We now verify that GSW satisfies the additional properties required by $\pfhe$.

\paragraph{GSW Satisfies $\pfhe$ Properties.}

    \begin{itemize}[leftmargin=*, label=-]
        \item $\eval$ is deterministic by the construction of GSW.

        \item (Pseudorandom $\pk$, $\ct$.)  This is just the built-in security of GSW shown via the LWE hardness and the Left-over Hash lemma. Assuming the $(\calT, \epsilon)$-hardness of $\LWE_{g, h, q, \chi}$, we have:
        
        \[\left\{ (\mat{A}, \mat{A}\cdot \mat{R} + \mu \cdot \mat{G}):  \begin{aligned}
        \mat{A} &\gets \pkgen \\
        \mat{R} &\gets \{0,1\}^{h \times h} \\
        \mu &\in \{0,1\}
        \end{aligned}\right\} \approx_{(\calT, \epsilon)} \left\{(\mat{U}_1, \mat{U}_2):  \begin{aligned}
        \mat{U}_1 &\gets \Z_q^{(g+1) \times h} \\ \mat{U}_2 &\gets \Z_q^{(g+1) \times h}
        \end{aligned}\right\}.\]

        \item ($\Z$ Linear Homomorphism.) Consider the following arbitrary linear combination, where each ciphertext $\ct_i = \mat{A}\cdot \mat{R}_i + \mu_i \cdot \mat{G}$ for $i \in [k]$:
        \[
            \calM = \sum_{i \in [k]} x_i\cdot \ct_i
            = \mat{A}\cdot ( \sum_{i \in [k]} x_i\mat{R}_i) + (\sum_{i \in [k]} x_i\mu_i) \cdot \mat{G}
        \]

        For $k\cdot \max_{i \in [k]} x_i \ll q$ , it can be deterministically decrypted using the normal GSW decryption for integer messages. For completeness, we briefly describe $\ldec$ as follows:

        \begin{itemize}
            \item $\ldec(\pk, \sk, \calM)$: Compute
            \begin{align*}
                \sk \cdot \calM &= [-\overline{\vec{s}}^\top \mid 1] \cdot \mat{A}\cdot ( \sum_{i \in [k]} x_i\mat{R}_i) + (\sum_{i \in [k]} x_i\mu_i)\cdot  [-\overline{\vec{s}}^\top \mid 1] \cdot \mat{G} \\
                &= \vec{e}^\top\cdot ( \sum_{i \in [k]} x_i\mat{R}_i) + (\sum_{i \in [k]} x_i\mu_i)\cdot  [-\overline{\vec{s}}^\top \mid 1] \cdot \mat{G} \\
                &\approx  (\sum_{i \in [k]} x_i\mu_i)\cdot  [-\overline{\vec{s}}^\top \mid 1] \cdot \mat{G}
            \end{align*}
            Since it is easy to solve LWE with respect to $\mat{G}$, we can recover  $(\sum_{i \in [k]} x_i\mu_i)\cdot  [-\overline{\vec{s}}^\top \mid 1]$ and hence $\sum_{i \in [k]} x_i\mu_i$ (modulo q).
        \end{itemize}
        
    \end{itemize}

    \paragraph{Efficiency Based on Subexponential LWE.}
    We choose the lattice dimension parameter $g = \secparam$, the modulus $q$ and the noise $\sigma$ with subexponential modulus-to-noise ratio, and the sample complexity $h \in \Theta(g \cdot \log q)$. The subexponential hardness of LWE requires the below indistinguishability to hold for adversaries of size $\calT(\secparam) = 2^{\secparam^\beta}$ for some constant $\beta > 0$, with subexponentially small distinguishing advantage $\epsilon(\secparam) = 2^{-\secparam^c}$ for some constant $c > 0$ denoted by $\approx_\epsilon$:
    \[\left\{ \begin{aligned}(\mat{A}, \vec{s} \cdot \mat{A} + \vec{e}) &\mod q:\\
        \mat{A} &\gets \Z^{g \times h}_q\\
        \vec{s} &\gets \Z^{1 \times g}_q\\
        \vec{e} &\gets \chi^{1 \times h}_\sigma 
    \end{aligned}\right\} \approx_\epsilon \left\{(\mat{A}, \vec{u}):\begin{aligned}
        \mat{A} &\gets \Z^{g \times h}_q\\
        \vec{u} &\gets \Z^{1 \times g}_q\\
    \end{aligned}\right\}\]

We set the security parameter be $\secparam = \log^{r} n$ for some constant $r > 1/\beta$ and $r > 1/c$. This satisfies the relationship between functions required by the Robustness Theorem \ref{thm:stream}: $\calT(\log^{r} n) = 2^{{(\log n)^{r\beta} }} = \Omega(\poly n)$ and
$\epsilon(\log^{r} n) = 2^{{-(\log n)^{rc} }}\leq \negl(n)$ for some negligible functions $\negl$.

Given these parameters, using GSW to implement $\pfhe$ and Theorem \ref{thm:rrec1} to implement $\sa_0$, we get the following result:

\begin{proposition}\label{prop:GSWcomplexity}
    Let $\calR(\A)$ denote the runtime of an algorithm $\A$. Assuming the above subexponential hardness of LWE, we have the following complexity regarding $\sa$ in Construction \ref{const1}:
    \begin{enumerate}[topsep=0.7em]
        \item $(\tilde{\pk}, \tilde{\ct_1}, \cdots, \tilde{\ct_{\ceil{\log n}}})$ and $sketch$ can be stored using $\OO(1)$ total bits of space. 

        \item $\calR(\sa.\update) = \OO(1) + \calR(\sa_0.\update)$.
        
        \item $\calR(\sa.\report) = \OO(k) + \calR(\sa_0.\report)$.
        

        
    \end{enumerate}
\end{proposition}

\begin{proof} Recall that computations in $\sa$ are modulo $q$ with $q \in \poly(n)$ and $q \gg n\cdot N$, where $N \in \poly(n)$ is the bound of the stream. Thus all entries have bit complexity in $\calO(\log n)$.
\begin{enumerate}[leftmargin=*]
    \item $\tilde{\pk}, sketch$, and $\tilde{\ct}_i$ for $i \in \ceil{\log n}$ are all $\calO(g) \times \calO(h)$ matrices over $\Z_q$. For $g = \secparam \in \poly(\log n)$ and $h \in \Theta(g \cdot \log q)$, they can be stored using $\calO(\poly(\log n)) = \OO(1)$ total bits of space.

    \item    The second bullet point considers the following operation:
    \begin{itemize}[topsep=0.3em]
        \item Evaluate $\hat{\ct_{i_t}} \gets \pfhe.\eval(\tilde{\pk}, C_{i_t}, \tilde{\ct_1}, \cdots, \tilde{\ct_\ell})$. Then update $sketch$ as $sketch \gets sketch + \Delta_t \cdot \hat{\ct_{i_t}}$.
    \end{itemize}

    Updating $sketch$ takes $\OO(1)$ time. For evaluation, recall that $C_{i_t}$ is a point function checking whether the input bits equal $i_t$. This can be naively implemented using $\calO(\log n)$ arithmetic operations, each of which runs in $\poly \log n$ time as described in $\eval_{+}$ and $\eval_{\times}$ above. Thus the total runtime is $\OO(1)$. 

    \item  The third bullet point considers the following operation:
    \begin{itemize}[topsep=0.3em]
        \item (If $\lVert \vec{x}'\rVert_0 > k$, output $\bot$. Otherwise:) Evaluate $\hat{\ct_i} \gets \pfhe.\eval(\tilde{\pk}, C_{i}, \tilde{\ct_1},$ $ \cdots, \tilde{\ct_\ell})$ for $i \in [n]$. Then compute $hash \gets \sum_{i \in [n]} \hat{\ct_i}\cdot \vec{x'}_i$.
    \end{itemize}

    Observe that although the construction is described as evaluating $n$-many $\hat{\ct_i}$, one for each coordinate of $\vec{x}'$, at this line of the algorithm $\vec{x}'$ is guaranteed to have at most $k$ non-zero entries. Thus $k$ evaluations (corresponding to non-zero coordinates of $\vec{x}'$) suffice for computing the linear combination $hash$. Each evaluation is $\OO(1)$, so the total runtime is $\OO(k)$.
\end{enumerate}\end{proof}

\paragraph{Efficiency of $\pfhe$ Based on Polynomial LWE.}

Under the polynomial LWE assumption, we set $\lambda = n^w$ for an arbitrarily small constant $w>0$. The $(\calT, \epsilon)$-security then reduces to the standard definition in terms of PPT adversaries and negligible advantage. We again choose $g = \secparam$ and $h \in \Theta(g \cdot \log q)$. The complexity is:

\begin{proposition}
    Let $\calR(\A)$ denote the runtime of an algorithm $\A$. Assuming the polynomial hardness of LWE, we have the following complexity regarding $\sa$ in Construction \ref{const1}:
    \begin{enumerate}[topsep=0.7em]
        \item $(\tilde{\pk}, \tilde{\ct_1}, \cdots, \tilde{\ct_{\ceil{\log n}}})$ and $sketch$ can be stored in $\OO(n^{w_1})$ bits of space. 

        \item $\calR(\sa.\update) = \OO(n^{w_2}) + \calR(\sa_0.\update)$.
        
        \item $\calR(\sa.\report) = \OO(k\cdot n^{w_3}) + \calR(\sa_0.\report)$.
    \end{enumerate}

    where $w_1, w_2$, and $w_3$ are arbitrarily small positive constants.
\end{proposition}

\begin{proof} The proof is exactly the same as the proof for Proposition \ref{prop:GSWcomplexity}, except that now the dimension of the ciphertext matrices blows up from $\poly \log n$ to $n^w$ for an arbitrarily small $w > 0$.

\end{proof}

\subsection{Complexity of Construction \ref{const1}}

Putting together Sections \ref{sec:instant-sa0} and \ref{sec:instant-pfhe}, we summarize the complexity of our construction in Table \ref{tab:stream} below.

\input{content/table2}

%% file: content/table2.tex
\begin{table*}[h!]
    \caption{A summary of the time and space complexities of the $\sa$ scheme in Construction \ref{const1}, given different instantiations of $\sa_0$ and $\pfhe$. $c$ denotes an arbitrarily small positive constant.}\label{tab:stream}
    \begin{center}
    \begin{small}
    \begin{sc}
    \begin{tabular}{lcccccr}
    \toprule
    Relaxed Rec. Scheme & LWE Assumption & Space & Update Time & Report Time\\
    \midrule

     Theorem \ref{thm:rrec1}  & subexp & $\tilde{\mathcal{O}}(k)$ & $\tilde{\mathcal{O}}(1)$  & {$\tilde{\mathcal{O}}(k^{1+c})$} \\

     Theorem \ref{thm:rrec2}  & subexp & $\tilde{\mathcal{O}}(k)$ & $\tilde{\mathcal{O}}(1)$  & {$\tilde{\mathcal{O}}(k^2)$}  \\

     Theorem \ref{thm:rrec1}   & poly & $\tilde{\mathcal{O}}(k+n^{c})$ & $\tilde{\mathcal{O}}(n^{c})$  & {$\tilde{\mathcal{O}}(k \cdot (k^{c} + n^{c}))$} \\

     Theorem \ref{thm:rrec2}  & poly & $\tilde{\mathcal{O}}(k+n^{c})$ & $\tilde{\mathcal{O}}(n^{c})$  & {$\tilde{\mathcal{O}}(k \cdot (k + n^{c}))$}  \\
    \bottomrule
    \end{tabular}
    \end{sc}
    \end{small}
    \end{center}
    \vskip -0.1in
\end{table*}

%% file: content/distributed.tex
\section{Distributed $K$-Sparse Recovery from Ring LWE}\label{sec:dis}

\subsection{Distributed Computation}\label{sec:dis-model}

In the distributed computation model, the dataset is 
 split over multiple servers. The goal is to design a communication protocol for servers to collectively compute a function on the aggregated dataset. In this work, we consider the message-passing coordinator model (see, e.g., \cite{phillips2015lower}), where there are $T$ servers and a single central coordinator. Each server $s_t: t\in[T]$ receives a partition of data in the form of an $n$-dimensional vector $\vec{x}_t$, and communicates with the central coordinator through a two-way, single-round
communication channel. At the end of the communication, the coordinator needs to compute or approximate a function $f_n(\vec{x})$ of the
aggregated vector $\vec{x} = \sum_{t \in [T]} \vec{x}_t$. We aim to design protocols that
minimize the total amount of communication through the channel and speed up the computation at each server and the coordinator. See also \cite{10.1145/2213977.2214063} and the references therein on the distributed model. 

Similar to the streaming model, we focus on designing a robust protocol in a white-box adversarial distributed setting, where the dataset is generated and partitioned across each server by a white-box adversary who monitors the internal state of the coordinator and the servers. We define the model as follows:

\begin{definition}[Distributed Protocol]
    Given a query function ensemble $\mathcal{F} = \{f_n : \Z^n \to X\}_{n \in \N}$ for some domain $X$, a distributed protocol $\dpr$ that computes $\mathcal{F}$ contains a tuple of algorithms $\dpr = (\dpr.\setup, $ $\dpr.\ser, \dpr.\cor)$, with the following properties:

    \begin{itemize}
        \item $\dpr.\setup(n, \rho_0) \to \zeta$ : On input a parameter $n$, the coordinator samples randomness $\rho_0$ and uses $\rho_0$ to compute some state $\zeta$. It sends $\zeta$ to all servers.

        \item $\dpr.\ser(\zeta, \vec{x}_t, \rho_t) \to p_t$ : On input a state $\zeta$ and a data partition $\vec{x}_t \in \Z^n$, the $t$-th server samples randomness $\rho_t$ and uses it to compute a packet $p_t$. It sends $p_t$ back to the coordinator.

        \item $\dpr.\cor(p_1, \cdots, p_T, \rho_c) \to r$: On input $T$ packets $p_1, \cdots, p_T$, the coordinator samples randomness $\rho_c$ and computes a query response $r \in X$.
        
        A response $r$ is said to be correct if $r = f_n(\vec{x})$ where $\vec{x} = \sum_{t \in [T]}\vec{x}_t$.
    \end{itemize}
\end{definition}

As in the streaming setting, we again assume bounded integer inputs, such that an arbitrary subset sum of $\{\vec{x}_t: t\in[T]\}$ is in $[-N, N]^n$ for some $N \in \poly(n)$.

\begin{definition}
    We say that a distributed protocol $\dpr$ is white-box adversarially robust if for all dimension parameters $n \in \N$, the following holds: For any PPT adversary $\A$,
    the following experiment $\expt_{\A, \dpr}(n)$ outputs $1$ with probability at most $\negl(n)$: \vspace{0.5em}
    
    $\expt_{\A, \dpr}(n): $
    \begin{enumerate}
        \item The challenger and $\A$ agree on a query function ensemble $\mathcal{F} = \{f_n : \Z^n \to X\}_{n \in \N}$ for some domain $X$. On input a dimension parameter $n$, let $f_n \in \mathcal{F}$ be the query function of the experiment.

        \item The challenger samples all randomness $\rho_0$, $\rho_t : t \in [T]$, and $\rho_c$ for the servers and the coordinator. It then provides $\rho_0, \rho_1, \cdots, \rho_T$ and $\rho_c$ to $\A$.

        \item $\A$ generates $T$ partitions of data $\vec{x}_1, \cdots, \vec{x}_T \in \Z^n$ and provides them to the challenger.

        \item The challenger executes the protocol to sequentially compute the following:
        
         \begin{enumerate}[label=(\roman*)]
             \item $\zeta \gets \dpr.\setup(n, \rho_0)$.

             \item $p_t \gets \dpr.\ser(\zeta, \vec{x}_t, \rho_t)$ for all $t \in [T]$.

             \item $r \gets \dpr.\cor(p_1, \cdots, p_T, \rho_c)$.
         \end{enumerate}

        \item The experiment outputs $1$ if and only if at some time $r \neq f_n(\vec{x})$.
    \end{enumerate}
\end{definition}

\begin{remark}
    We remark on the connection between the streaming model and the distributed model. There exists a well-known reduction from a streaming algorithm to a distributed protocol under the oblivious setting: Each server runs $\sa.\update$ on the same initial state $\zeta \gets \sa.\setup$, using all entries of $\vec{x}_t$ (its data partition) as updates. It then sends the resulting sketch to the coordinator. As long as the sketches are linear, the coordinator can run $\sa.\report$ on the sum of all sketches to compute a query response.
    \end{remark}

    One may use the same recipe to construct a {\it WAR} distributed protocol from a WAR streaming algorithm. In general, the robustness is retained by such a reduction if $\sa.\update$ and $\sa.\report$ are deterministic. (Otherwise, the adversaries in the experiment $\expt_{\A, \dpr}$ would have additional views to all ``future'' randomness, while adversaries in the streaming experiment $\expt_{\A, \sa}$ only monitor the past randomness already used by $\sa$.) Our Construction \ref{const1} satisfies this condition.

    \subsection{Distributed Protocol Construction}

    For completeness, we briefly describe the distributed protocol for $k$-Sparse Recovery as follows, which essentially follows the same route as our streaming Construction \ref{const1}.

    \begin{const}\label{const2}

    Similar to Construction \ref{const1}, let $\pfhe$ be a pseudorandom $\fhe$ scheme. And let $\sa_0$ be a streaming algorithm for the Relaxed $k$-Sparse Recovery problem. We assume that $\sa_0$ performs a deterministic linear measurement $\Phi_n$ on an arbitrary length-$n$ input vector, as described in the previous instantiation Section \ref{sec:instant-sa0}.
    \begin{itemize}[leftmargin=*,topsep=1em,itemsep=1em]
        \item $\dpr.\setup(n)$ :        Let $\lambda = \calS(n)$ be the security parameter of $\pfhe$. Sample $\tilde{\pk} \gets \D_{\tilde{\pk}(\lambda)}$ and $\tilde{\ct_1}, \tilde{\ct_2}, \cdots, \tilde{\ct_{\ceil{\log n}}} \gets \D_{\tilde{\ct}(\lambda)}$. Send $\zeta:= (\tilde{\pk}, \tilde{\ct_1}, \cdots, \tilde{\ct_{\log n}})$ to all servers

        \item $\dpr.\ser(\zeta, \vec{x_t})$ :
        Given $\zeta:= (\tilde{\pk}, \tilde{\ct_1}, \cdots, \tilde{\ct_{\ceil{\log n}}})$, evaluate $\hat{\ct_i} \gets \pfhe.$ $\eval(\tilde{\pk}, C_{i}, \tilde{\ct_1}, \cdots,$ $\tilde{\ct_{\ceil{\log n}}})$ for $i \in [n]$, where $C_i$ is the same point function as in Construction \ref{const1}. 
        
        Compute $\Phi_n(\vec{x_t})$ and $sketch_t \gets \sum_{i \in [n]} \hat{\ct_i}\cdot (\vec{x_t})_i$, where $(\vec{x_t})_i$ denotes the $i$-th coordinate of $\vec{x_t}$. Send $p_t := (sketch_t, \Phi_n(\vec{x_t}))$ to the coordinator.

        \item $\dpr.\cor(p_1, \cdots, p_T)$: Let $\vec{x}' \gets \sa_0.\report(\sum_{t \in  [T]} \Phi_n(\vec{x_t}))$, i.e., the output of running the reconstruction algorithm of $\sa_0$ on the sum of all measurement results. If $\lVert \vec{x}'\rVert_0 > k$, output $\bot$.

        Otherwise, evaluate $\hat{\ct_i} \gets \pfhe.\eval(\tilde{\pk}, C_{i}, \tilde{\ct_1}, \cdots, \tilde{\ct_{\ceil{\log n}}})$ for $i \in [n]$. Then compute $hash \gets \sum_{i \in [n]} \hat{\ct_i}\cdot \vec{x'}_i.$ Output $\vec{x}'$ if $hash = \sum_{t \in [T]} sketch_t$, and output $\bot$ otherwise.

        \end{itemize}
    \end{const}

    \subsection{Efficiency}
    In this section, we focus on showing that we can use BV \cite{BV11F}, a Ring-LWE based $\fhe$ scheme, to achieve near-linear processing time on $\dpr.\ser$ under the {\it polynomial} ring-LWE assumption.

    As a warm-up, we state the complexity of Constrution \ref{const2} using GSW to implement $\pfhe$:
    
\begin{proposition}
    Let $\calR(\A)$ denote the runtime of an algorithm or an evaluation $\A$, we have the following complexity regarding $\dpr$ in Construction \ref{const2}:
    \begin{itemize}[leftmargin=*]
    \item Assuming the subexponential hardness of LWE:
    \begin{enumerate}[topsep=0.3em]
        \item The total bits of communication between an arbitrary server and the coordinator is $\OO(1)$ + $size(\Phi_n(\vec{x}))$.
        
        \item $\calR(\dpr.\cor) = \OO(k) + \calR(\sa_0.\report)$.
        
        \item $\calR(\dpr.\ser) = \OO(k) + \calR(\Phi_n(\vec{x}))$.
    \end{enumerate}

        \item Assuming the polynomial hardness of LWE: :
    \begin{enumerate}[topsep=0.3em]
        \item The total bits of communication between an arbitrary server and the coordinator is $\OO(n^{w_1})$ + $size(\Phi_n(\vec{x}))$.
        
        \item $\calR(\dpr.\cor) = \OO(k\cdot n^{w_2}) + \calR(\sa_0.\report)$.
        
        \item $\calR(\dpr.\ser) = \OO(k\cdot n^{w_3}) + \calR(\Phi_n(\vec{x}))$.
    \end{enumerate}
\end{itemize}    

    where $\vec{x} \in [-N, N]^n$ is arbitrary, and $w_1, w_2,$ and $w_3$ are arbitrarily small positive constants.
\end{proposition}

\begin{proof} 
    The communication complexity and the runtime of $\dpr.\cor$ are the same as the bit complexity and reporting time of $\sa$ in Construction \ref{const1} (except for some negligible addition operations). The runtime of $\dpr.\ser$ is roughly $n \cdot \calR(\sa.\update)$, because every server has to perform $n$ evaluations, one for each coordinate of its data partition.

\end{proof}

\paragraph{Efficiency Based on Ring-LWE.} Before discussing the efficiency advantage, we briefly review some important features of BV to show that it satisfies a variant of the $\pfhe$ definition. See \cite{BV11F} for details of the BV construction.

Consider a ring $R := \frac{\mathbb{Z}[x]}{x^{g}+1}$ where $g$ is the degree parameter. Let $R_q$ denote $R/qR$ for a prime modulus $q$, and let $R_t$ be the message space for some prime $t \in \Z_q^*$ and $t \ll q/g$:

\begin{itemize}[leftmargin=*, topsep=0.2em]
    \item The public keys of BV are two-dimensional vectors in $R^2_q$, which are computationally indistinguishable from uniformly random samples from $R^2_q$ under the Ring-LWE assumption.
    
    The ciphertexts are also vectors over $R_q$ that are indistinguishable from uniform random.

    \item BV satisfies an analogy of $\Z$ Linear Homomorphism, with coefficients of linear combinations $\calM = \sum_{i \in [k]}x_i\cdot \ct_i$ being small norm ring elements $x_i \in R_t$ as opposed to integer scalars. As long as the linear combination of corresponding messages ($\calM_{msg} = \sum_{i \in [k]}x_i\cdot m_i$, where $m_i \in R_t$ is the message encrypted by $\ct_i$) still lies in $R_t$, running the standard BV decoding algorithm on $\calM$ output the correct message.

    \item $\eval$ and $\dec$ are deterministic by the construction of BV.
\end{itemize}

We choose parameters for our Ring-LWE based construction as follows: Let $\secparam = n^w$ for an arbitrarily small constant $w >0$. Let the dimension $g = \secparam$, the modulus of the message space $t > n \cdot N$, where $N \in \poly(n)$ is the bound of the input data. The modulus $q$ is chosen with subexponential modulus-to-noise ratio and $q \gg g \cdot t$ the Ring-LWE assumption states that for any $h \in 
 \poly(\secparam)$, it holds that 
\[\left\{(a_i, a_i\cdot s  + e_i): 
        e_i \gets \chi 
    \right\}_{i \in [h]} \approx_c \left\{(a_i, u_i):\begin{aligned}
        u_i &\gets R_q\\
    \end{aligned}\right\}_{i \in [h]}\]

    where $s$ is sampled from the noise distribution $\chi$ and $a_i$ are uniform in $R_q$.

We have the following complexity:

\begin{proposition}
    Let $\calR(\A)$ denote the runtime of an algorithm or an evaluation $\A$. Assuming the polynomial hardness of Ring-LWE, a variant of $\dpr$ in Construction \ref{const2} has the following complexity:

    \begin{enumerate}[topsep=0.3em]
        \item The total bits of communication between an arbitrary server and the coordinator is $\OO(n^{w_1})$ + $size(\Phi_n(\vec{x}))$.
        
        \item $\calR(\dpr.\cor) = \OO(n) + \calR(\sa_0.\report)$.
        
        \item $\calR(\dpr.\ser) = \OO(k\cdot n^w_2) + \calR(\Phi_n(\vec{x}))$.
    \end{enumerate}

    where $\vec{x} \in [-N, N]^n$ is arbitrary, and $w_1, w_2$ are arbitrarily small positive constants.
\end{proposition}

\begin{proof}
    In comparison to the complexity under polytime standard LWE, we save a $n^{\Theta(1)}$ factor in the runtime of $\dpr.\ser$. This is due to speeding up the following two steps:
    \begin{enumerate}
        \item Evaluate $\hat{\ct_i} \gets \pfhe.\eval(\tilde{\pk}, C_{i}, \tilde{\ct_1}, \cdots, \tilde{\ct_{\ceil{\log n}}})$ for all $i \in [n]$

        \item Compute $\sum_{i \in [n]} \hat{\ct_i}\cdot \vec{x}_i$
    \end{enumerate}

    The idea is to represent the messages as ring elements over $R_q$ as opposed to individual scalars, and apply fast Fourier transforms to efficiently multiply them. 
    
    In a variant of Construction \ref{const2}, instead of performing $n$ evaluations one for each coordinate $i \in [n]$, we partition the dimension $n$ into $n^{1-w}$ chunks. Each chunk is of length $n^w$, where $g = n^w$ is the degree parameter of the ring for an arbitrarily small constant $w >0$. Such that a length-$n$ vector $\vec{x}$ is partitioned into $n^{1-w}$ consecutive segments, each can be represented as a (small-normed in comparison to $q$) ring element $x_j \in R_q$ for $j \in [n^{1-w}]$.

    We also define a new set of functions $C_j'$ as follows:
    \[C'_{j}(\mu_1, \cdots, \mu_{\ceil{\log n}}) = \]
    \[\begin{cases}
            1 & \text{if $(\mu_1, \cdots, \mu_{\ceil{\log n}})$ represents a value in $[jn^w, (j+1)n^w)$} \\
            0 & \text{otherwise}.
        \end{cases}\]

    We evaluate $\hat{\ct_j'} \gets \pfhe.\eval(\tilde{\pk}, C'_{j}, \tilde{\ct_1}, \cdots, \tilde{\ct_{\ceil{\log n}}})$ for all $j \in [n^{1-w}]$ to replace the above step 1. The total evaluation time is $\calO(n^{1-w} \cdot n^{w}\poly \log (n^w)) = \OO(n)$, since two ring elements can be multiplied in $\calO(n^{w}\poly \log (n^w))$ time using the FFT. Then the linear combination in the above step 2 is replaced by $\sum_{j \in [n^{1-w}]} \hat{\ct'_j}\cdot x_j$, which similarly takes $\OO(n)$ time to compute.

    Lastly, we check that given two vectors $\vec{x}$ and $\vec{x'}$ differ at a coordinate $m$, when $\ct_1, \cdots, \ct_{\ceil{\log n}}$ encrypt the bits of $m$,  we still have \[sketch = \sum_{j \in [n^{1-w}]} \hat{\ct'_j}\cdot x_j \neq \sum_{j \in [n^{1-w}]} \hat{\ct'_j}\cdot x'_j = hash.\]

    because they can be decrypted into distinct ring elements. This concludes that the above variant of Construction \ref{const2} remains correct while achieving near-linear runtime of $\dpr.\cor$ and $\dpr.\ser$. 
    
\end{proof}

%% file: content/matrixandtensor.tex
\section{Low-Rank Matrix and Tensor Recovery}\label{sec:matrixTensor}

The low-rank matrix and tensor recovery problems are defined as follows:

\begin{definition}[Rank-$k$ Matrix Recovery]
    Given a rank parameter $k$ and an input matrix $\mat{X} \in \Z^{n \times m}$, output $\mat{X}$ if $rank(\mat{X}) \leq k$ and report $\bot$ otherwise. 
\end{definition}

Let $\otimes$ denote the outer product of two vectors. We use the following definition for tensor rank:

\begin{definition}[CANDECOMP/PARAFAC (CP) rank]
    For a tensor $X \in \R^{n^d}$, consider it to be the sum of $k$ rank-$1$ tensors:
    $X = \sum^k_{i=1}(x_{i1} \otimes x_{i2} \otimes \cdots  
    \otimes x_{id})$ where $x_{ij} \in \R^{n}$.
    The smallest number of rank-$1$
    tensors that can be used to express a tensor $X$ is then defined to be the rank of the tensor. 
\end{definition}

\begin{definition}[Rank-$k$ Tensor Recovery]
    Given a rank parameter $k$ and an input tensor $\mat{X} \in \Z^{n^d}$, output $\mat{X}$ if $CP-rank(\mat{X}) \leq k$ and report $\bot$ otherwise.
\end{definition}

The low-rank matrix and tensor recovery problems are natural generalizations of the $k$-sparse vector recovery problems. Using the same framework as the vector case, they can be solved by maintaining two hashes: one recovers the matrix or tensor but potentially produces garbage output, and the other uses a cryptographic tool to verify whether the output matches the input data. These are natural generalizations of the $k$-sparse vector recovery problems.

In \cite{fw23}, the recovery hash is measured by a random Gaussian matrix, and the verification hash is measured by an SIS matrix. Both measurements are assumed to be heuristically compressed by an idealized hash function. We instead replace the former with a deterministic measurement, which can be explicitly constructed and decoded in polynomial time \cite{10.1145/2213977.2213995}. The latter is replaced by the same FHE-based hashing as the one we used for the vector case in Construction \ref{const1}, which operates on the vectorization of the data matrix or tensor. Similar to the vector case, our construction does not require using an idealized hash function or storing a separate large state. 

In the streaming model, the complexities of our matrix and tensor recovery algorithms are as follows:

\begin{proposition}
    Let $\calR(\A)$ denote the runtime of an algorithm or an evaluation $\A$. Assuming the polynomial hardness of LWE,
    \begin{itemize}[topsep=0.3em]
        \item There exists a WAR streaming algorithm $\sa$ for the rank-$k$ matrix recovery problem using $\OO(nk)$ bits of space.
        
        \item There exists a WAR streaming algorithm $\sa$ for the rank-$k$ tensor recovery problem using $\OO(nk^{\ceil{\log d}})$ bits of space.
    \end{itemize}

\end{proposition}

\begin{proof} 
    The recovery scheme for matrices and tensors takes $\OO(nk)$ and $\OO(nk^{\ceil{\log n}})$ bits of space, respectively \cite{10.1145/2213977.2213995}. Under the polynomial hardness assumption of LWE, we set $\secparam = n^c$ for an arbitrarily small positive constant $c$. Taking the sum of the two terms produces the claimed complexities.
\end{proof}

In addition, the above streaming algorithms can be converted into WAR distributed protocols for matrix and vector recovery, following the same roadmap as in Section \ref{sec:dis}.

\begin{proposition}
    Let $\calR(\A)$ denote the runtime of an algorithm or an evaluation $\A$. Assuming the polynomial hardness of LWE: :
    \begin{itemize}[topsep=0.3em]
        \item There exists a WAR distributed protocol $\dpr$ for the rank-$k$ matrix recovery problem, using $\OO(nk)$ bits of communication between each server and the coordinator, and $\OO(n^2k)$ processing time at each server.
        
        \item There exists a WAR streaming algorithm $\sa$ for the rank-$k$ tensor recovery problem, using $\OO(nk^{\ceil{\log d}})$ bits of communication between each server and the coordinator.
    \end{itemize}
\end{proposition}

\begin{proof}
    The total number of bits of communication between each server and the coordinator equals the size of the state in the above streaming algorithms. For low-rank matrix recovery, the deterministic scheme performs $\OO(nk)$ measurements using $n$-sparse matrices, and the runtime for computing the FHE-based hash is comparably small. This gives $\OO(n^2k)$ processing time at each server.
\end{proof}

We remark that the rank-$k$ matrix recovery algorithm in the previous work \cite{fw23} requires $\OO(n^3k)$ processing time in the distributed setting, which our construction improves upon.

%% file: main.bbl
\newcommand{\etalchar}[1]{$^{#1}$}
\begin{thebibliography}{HQWW22}

\bibitem[ABD{\etalchar{+}}21]{AlonBDMNY21}
Noga Alon, Omri Ben{-}Eliezer, Yuval Dagan, Shay Moran, Moni Naor, and Eylon Yogev.
\newblock Adversarial laws of large numbers and optimal regret in online classification.
\newblock In {\em {STOC} '21: 53rd Annual {ACM} {SIGACT} Symposium on Theory of Computing}, pages 447--455, 2021.

\bibitem[ABJ{\etalchar{+}}22]{10.1145/3517804.3526228}
Mikl\'{o}s Ajtai, Vladimir Braverman, T.S. Jayram, Sandeep Silwal, Alec Sun, David~P. Woodruff, and Samson Zhou.
\newblock The white-box adversarial data stream model.
\newblock In {\em Proceedings of the 41st ACM SIGMOD-SIGACT-SIGAI Symposium on Principles of Database Systems}, PODS '22, pages 15--27, New York, NY, USA, 2022. Association for Computing Machinery.

\bibitem[ABS23]{ABS}
A.~Almeida, S.~Brás, and et~al. Sargento, S.
\newblock Time series big data: a survey on data stream frameworks, analysis and algorithms.
\newblock {\em J Big Data}, 2023.

\bibitem[ACSS21]{AttiasCSS21}
Idan Attias, Edith Cohen, Moshe Shechner, and Uri Stemmer.
\newblock A framework for adversarial streaming via differential privacy and difference estimators.
\newblock {\em CoRR}, abs/2107.14527, 2021.

\bibitem[AMS96]{10.1145/237814.237823}
Noga Alon, Yossi Matias, and Mario Szegedy.
\newblock The space complexity of approximating the frequency moments.
\newblock In {\em Proceedings of the Twenty-Eighth Annual ACM Symposium on Theory of Computing}, STOC '96, pages 20--29, New York, NY, USA, 1996. Association for Computing Machinery.

\bibitem[AR05]{Aharonov2005LatticePI}
Dorit Aharonov and Oded Regev.
\newblock Lattice problems in np and conp.
\newblock {\em J. ACM}, 52:749--765, 2005.

\bibitem[BEO22]{Ben-EliezerEO22}
Omri Ben{-}Eliezer, Talya Eden, and Krzysztof Onak.
\newblock Adversarially robust streaming via dense-sparse trade-offs.
\newblock In {\em 5th Symposium on Simplicity in Algorithms, {SOSA}}, 2022.
\newblock (to appear).

\bibitem[BGV12]{BGV12}
Zvika Brakerski, Craig Gentry, and Vinod Vaikuntanathan.
\newblock (leveled) fully homomorphic encryption without bootstrapping.
\newblock In {\em Innovations in Theoretical Computer Science 2012, Cambridge, MA, USA, January 8-10, 2012}, pages 309--325, 2012.

\bibitem[BHM{\etalchar{+}}05]{BrownHMPRS05}
Paul Brown, Peter~J. Haas, Jussi Myllymaki, Hamid Pirahesh, Berthold Reinwald, and Yannis Sismanis.
\newblock Toward automated large-scale information integration and discovery.
\newblock In {\em Data Management in a Connected World, Essays Dedicated to Hartmut Wedekind on the Occasion of His 70th Birthday}, pages 161--180, 2005.

\bibitem[BHM{\etalchar{+}}21]{BravermanHMSSZ21}
Vladimir Braverman, Avinatan Hassidim, Yossi Matias, Mariano Schain, Sandeep Silwal, and Samson Zhou.
\newblock Adversarial robustness of streaming algorithms through importance sampling.
\newblock In {\em Advances in Neural Information Processing Systems 34: Annual Conference on Neural Information Processing Systems, NeurIPS}, 2021.
\newblock (to appear).

\bibitem[BIR08]{4797556}
R.~Berinde, P.~Indyk, and M.~Ruzic.
\newblock Practical near-optimal sparse recovery in the l1 norm.
\newblock In {\em 2008 46th Annual Allerton Conference on Communication, Control, and Computing}, pages 198--205, 2008.

\bibitem[BJWY21]{Ben-EliezerJWY21}
Omri Ben{-}Eliezer, Rajesh Jayaram, David~P. Woodruff, and Eylon Yogev.
\newblock A framework for adversarially robust streaming algorithms.
\newblock {\em {SIGMOD} Rec.}, 50(1):6--13, 2021.

\bibitem[BLV19]{BLV19}
Elette Boyle, Rio LaVigne, and Vinod Vaikuntanathan.
\newblock Adversarially robust property-preserving hash functions.
\newblock In Avrim Blum, editor, {\em 10th Innovations in Theoretical Computer Science Conference, {ITCS} 2019, January 10-12, 2019, San Diego, California, {USA}}, volume 124 of {\em LIPIcs}, pages 16:1--16:20. Schloss Dagstuhl - Leibniz-Zentrum f{\"{u}}r Informatik, 2019.

\bibitem[BV11]{BV11F}
Zvika Brakerski and Vinod Vaikuntanathan.
\newblock Fully homomorphic encryption from ring-lwe and security for key dependent messages.
\newblock In Phillip Rogaway, editor, {\em Advances in Cryptology -- CRYPTO 2011}, pages 505--524, Berlin, Heidelberg, 2011. Springer Berlin Heidelberg.

\bibitem[BY20]{Ben-EliezerY20}
Omri Ben{-}Eliezer and Eylon Yogev.
\newblock The adversarial robustness of sampling.
\newblock In {\em Proceedings of the 39th {ACM} {SIGMOD-SIGACT-SIGAI} Symposium on Principles of Database Systems, {PODS}}, pages 49--62, 2020.

\bibitem[CCD{\etalchar{+}}03]{10.1145/872757.872857}
Sirish Chandrasekaran, Owen Cooper, Amol Deshpande, Michael~J. Franklin, Joseph~M. Hellerstein, Wei Hong, Sailesh Krishnamurthy, Samuel~R. Madden, Fred Reiss, and Mehul~A. Shah.
\newblock Telegraphcq: Continuous dataflow processing.
\newblock In {\em Proceedings of the 2003 ACM SIGMOD International Conference on Management of Data}, SIGMOD '03, page 668, New York, NY, USA, 2003. Association for Computing Machinery.

\bibitem[CG08]{10.1145/1353343.1353442}
Graham Cormode and Minos Garofalakis.
\newblock Streaming in a connected world: Querying and tracking distributed data streams.
\newblock In {\em Proceedings of the 11th International Conference on Extending Database Technology: Advances in Database Technology}, EDBT '08, page 745, New York, NY, USA, 2008. Association for Computing Machinery.

\bibitem[CGS22]{ChakrabartiGS22}
Amit Chakrabarti, Prantar Ghosh, and Manuel Stoeckl.
\newblock Adversarially robust coloring for graph streams.
\newblock In {\em 13th Innovations in Theoretical Computer Science Conference, {ITCS}}, 2022.
\newblock (to appear).

\bibitem[CH21]{ChanH21}
Timothy~M. Chan and Qizheng He.
\newblock More dynamic data structures for geometric set cover with sublinear update time.
\newblock In {\em 37th International Symposium on Computational Geometry, SoCG}, pages 25:1--25:14, 2021.

\bibitem[Cha10]{Chan10}
Timothy~M. Chan.
\newblock A dynamic data structure for 3-d convex hulls and 2-d nearest neighbor queries.
\newblock {\em J. {ACM}}, 57(3):16:1--16:15, 2010.

\bibitem[CJ19]{10.1145/3297715}
Graham Cormode and Hossein Jowhari.
\newblock Lp samplers and their applications: A survey.
\newblock {\em ACM Comput. Surv.}, 52(1), feb 2019.

\bibitem[CM05]{1561219}
G.~Cormode and S.~Muthukrishnan.
\newblock What's new: finding significant differences in network data streams.
\newblock {\em IEEE/ACM Transactions on Networking}, 13(6):1219--1232, 2005.

\bibitem[CZM{\etalchar{+}}18]{CubukZMVL18}
Ekin~Dogus Cubuk, Barret Zoph, Dandelion Man{\'{e}}, Vijay Vasudevan, and Quoc~V. Le.
\newblock Autoaugment: Learning augmentation policies from data.
\newblock {\em CoRR}, abs/1805.09501, 2018.

\bibitem[DJMS02]{DasuJMS02}
Tamraparni Dasu, Theodore Johnson, S.~Muthukrishnan, and Vladislav Shkapenyuk.
\newblock Mining database structure; or, how to build a data quality browser.
\newblock In {\em Proceedings of the 2002 {ACM} {SIGMOD} International Conference on Management of Data}, pages 240--251, 2002.

\bibitem[DSST89]{DriscollSST89}
James~R. Driscoll, Neil Sarnak, Daniel~Dominic Sleator, and Robert~Endre Tarjan.
\newblock Making data structures persistent.
\newblock {\em J. Comput. Syst. Sci.}, 38(1):86--124, 1989.

\bibitem[FK03]{FiatK03}
Amos Fiat and Haim Kaplan.
\newblock Making data structures confluently persistent.
\newblock {\em J. Algorithms}, 48(1):16--58, 2003.

\bibitem[FN85]{FLAJOLET1985182}
Philippe Flajolet and G.~{Nigel Martin}.
\newblock Probabilistic counting algorithms for data base applications.
\newblock {\em Journal of Computer and System Sciences}, 31(2):182--209, 1985.

\bibitem[FS12]{10.1145/2213977.2213995}
Michael~A. Forbes and Amir Shpilka.
\newblock On identity testing of tensors, low-rank recovery and compressed sensing.
\newblock In {\em Proceedings of the Forty-Fourth Annual ACM Symposium on Theory of Computing}, STOC '12, page 163–172, New York, NY, USA, 2012. Association for Computing Machinery.

\bibitem[FW23]{fw23}
Ying Feng and David~P. Woodruff.
\newblock Improved algorithms for white-box adversarial streams.
\newblock In {\em ICML}, 2023.

\bibitem[Gab07]{Gaber2007}
Mohamed~Medhat Gaber.
\newblock {\em Data Stream Processing in Sensor Networks}, pages 41--48.
\newblock Springer Berlin Heidelberg, Berlin, Heidelberg, 2007.

\bibitem[GG00]{GG00}
Oded Goldreich and Shafi Goldwasser.
\newblock On the limits of nonapproximability of lattice problems.
\newblock {\em J. Comput. Syst. Sci.}, 60(3):540--563, 2000.

\bibitem[GLPS17]{10.1145/3039872}
Anna~C. Gilbert, Yi~Li, Ely Porat, and Martin~J. Strauss.
\newblock For-all sparse recovery in near-optimal time.
\newblock {\em ACM Trans. Algorithms}, 13(3), mar 2017.

\bibitem[GSS14]{GoodfellowSS14}
Ian~J. Goodfellow, Jonathon Shlens, and Christian Szegedy.
\newblock Explaining and harnessing adversarial examples.
\newblock {\em CoRR}, abs/1412.6572, 2014.

\bibitem[GSW13]{C:GenSahWat13}
Craig Gentry, Amit Sahai, and Brent Waters.
\newblock Homomorphic encryption from learning with errors: Conceptually-simpler, asymptotically-faster, attribute-based.
\newblock In Ran Canetti and Juan~A. Garay, editors, {\em CRYPTO~2013, Part~I}, volume 8042 of {\em {LNCS}}, pages 75--92. Springer, Heidelberg, August 2013.

\bibitem[HKM{\etalchar{+}}20]{HassidimKMMS20}
Avinatan Hassidim, Haim Kaplan, Yishay Mansour, Yossi Matias, and Uri Stemmer.
\newblock Adversarially robust streaming algorithms via differential privacy.
\newblock In {\em Advances in Neural Information Processing Systems 33: Annual Conference on Neural Information Processing Systems, NeurIPS}, 2020.

\bibitem[HQWW22]{10.1145/3532189}
Yue Hu, Ao~Qu, Yanbing Wang, and Daniel~B. Work.
\newblock Streaming data preprocessing via online tensor recovery for large environmental sensor networks.
\newblock {\em ACM Trans. Knowl. Discov. Data}, 16(6), jul 2022.

\bibitem[HR07]{STOC:HavReg07}
Ishay Haviv and Oded Regev.
\newblock Tensor-based hardness of the shortest vector problem to within almost polynomial factors.
\newblock In David~S. Johnson and Uriel Feige, editors, {\em 39th ACM STOC}, pages 469--477. {ACM} Press, June 2007.

\bibitem[IEM18]{IlyasEM18}
Andrew Ilyas, Logan Engstrom, and Aleksander Madry.
\newblock Prior convictions: Black-box adversarial attacks with bandits and priors.
\newblock {\em CoRR}, abs/1807.07978, 2018.

\bibitem[Ind06]{Indyk06}
Piotr Indyk.
\newblock Stable distributions, pseudorandom generators, embeddings, and data stream computation.
\newblock {\em J. {ACM}}, 53(3):307--323, 2006.

\bibitem[Jaf11]{DeterministicCompressedSensing}
Sina Jafarpour.
\newblock {\em Deterministic Compressed Sensing}.
\newblock PhD thesis, Princeton University, 2011.

\bibitem[Kap04]{Kaplan04}
Haim Kaplan.
\newblock Persistent data structures.
\newblock In {\em Handbook of Data Structures and Applications}. Chapman and Hall/CRC, 2004.

\bibitem[KGB17]{KurakinGB17}
Alexey Kurakin, Ian~J. Goodfellow, and Samy Bengio.
\newblock Adversarial machine learning at scale.
\newblock In {\em 5th International Conference on Learning Representations, {ICLR}, Conference Track Proceedings}, 2017.

\bibitem[Kho04]{FOCS:Khot04a}
Subhash Khot.
\newblock Hardness of approximating the shortest vector problem in lattices.
\newblock In {\em 45th FOCS}, pages 126--135. {IEEE} Computer Society Press, October 2004.

\bibitem[KMNS21]{KaplanMNS21}
Haim Kaplan, Yishay Mansour, Kobbi Nissim, and Uri Stemmer.
\newblock Separating adaptive streaming from oblivious streaming using the bounded storage model.
\newblock In {\em Advances in Cryptology - {CRYPTO} 2021 - 41st Annual International Cryptology Conference, {CRYPTO}, Proceedings, Part {III}}, pages 94--121, 2021.

\bibitem[KNW10]{KNW10}
Daniel~M. Kane, Jelani Nelson, and David~P. Woodruff.
\newblock On the exact space complexity of sketching and streaming small norms.
\newblock In Moses Charikar, editor, {\em Proceedings of the Twenty-First Annual {ACM-SIAM} Symposium on Discrete Algorithms, {SODA} 2010, Austin, Texas, USA, January 17-19, 2010}, pages 1161--1178. {SIAM}, 2010.

\bibitem[KPTW23]{kacham2023pseudorandom}
Praneeth Kacham, Rasmus Pagh, Mikkel Thorup, and David~P. Woodruff.
\newblock Pseudorandom hashing for space-bounded computation with applications in streaming, 2023.

\bibitem[LCLS17]{LiuCLS17}
Yanpei Liu, Xinyun Chen, Chang Liu, and Dawn Song.
\newblock Delving into transferable adversarial examples and black-box attacks.
\newblock In {\em 5th International Conference on Learning Representations, {ICLR}, Conference Track Proceedings}, 2017.

\bibitem[Lig18]{Lightstone18}
Sam Lightstone.
\newblock Physical database design for relational databases.
\newblock In {\em Encyclopedia of Database Systems, Second Edition}. Springer, 2018.

\bibitem[LPR10]{Ec:LyuPeiReg10}
Vadim Lyubashevsky, Chris Peikert, and Oded Regev.
\newblock On ideal lattices and learning with errors over rings.
\newblock In Henri Gilbert, editor, {\em EUROCRYPT~2010}, volume 6110 of {\em {LNCS}}, pages 1--23. Springer, Heidelberg, May~/~June 2010.

\bibitem[LSO{\etalchar{+}}06]{10.1145/1140103.1140295}
Ashwin Lall, Vyas Sekar, Mitsunori Ogihara, Jun Xu, and Hui Zhang.
\newblock Data streaming algorithms for estimating entropy of network traffic.
\newblock {\em SIGMETRICS Perform. Eval. Rev.}, 34(1):145–156, jun 2006.

\bibitem[MMS{\etalchar{+}}18]{MadryMSTV18}
Aleksander Madry, Aleksandar Makelov, Ludwig Schmidt, Dimitris Tsipras, and Adrian Vladu.
\newblock Towards deep learning models resistant to adversarial attacks.
\newblock In {\em 6th International Conference on Learning Representations, {ICLR}, Conference Track Proceedings}, 2018.

\bibitem[MN21]{MenuhinN21}
Boaz Menuhin and Moni Naor.
\newblock Keep that card in mind: Card guessing with limited memory.
\newblock {\em CoRR}, abs/2107.03885, 2021.

\bibitem[Mut05]{TCS-002}
S.~Muthukrishnan.
\newblock Data streams: Algorithms and applications.
\newblock {\em Foundations and Trends® in Theoretical Computer Science}, 1(2):117--236, 2005.

\bibitem[Nis92]{Nisan92}
Noam Nisan.
\newblock Pseudorandom generators for space-bounded computation.
\newblock {\em Comb.}, 12(4):449--461, 1992.

\bibitem[PVZ15]{phillips2015lower}
Jeff~M. Phillips, Elad Verbin, and Qin Zhang.
\newblock Lower bounds for number-in-hand multiparty communication complexity, made easy, 2015.

\bibitem[RSW22]{RoghaniSW22}
Mohammad Roghani, Amin Saberi, and David Wajc.
\newblock Beating the folklore algorithm for dynamic matching.
\newblock In {\em 13th Innovations in Theoretical Computer Science Conference, {ITCS}}, 2022.
\newblock (to appear).

\bibitem[SAC{\etalchar{+}}79]{SelingerACLP79}
Patricia~G. Selinger, Morton~M. Astrahan, Donald~D. Chamberlin, Raymond~A. Lorie, and Thomas~G. Price.
\newblock Access path selection in a relational database management system.
\newblock In {\em Proceedings of the 1979 {ACM} {SIGMOD} International Conference on Management of Data}, pages 23--34. {ACM}, 1979.

\bibitem[SDNR96]{ShuklaDNR96}
Amit Shukla, Prasad Deshpande, Jeffrey~F. Naughton, and Karthikeyan Ramasamy.
\newblock Storage estimation for multidimensional aggregates in the presence of hierarchies.
\newblock In {\em VLDB'96, Proceedings of 22th International Conference on Very Large Data Bases}, pages 522--531, 1996.

\bibitem[SGV02]{SAGJ}
Alexander Szalay, Jim Gray, and Jan VandenBerg.
\newblock Petabyte scale data mining: Dream or reality?
\newblock {\em Proceedings of SPIE - The International Society for Optical Engineering}, 08 2002.

\bibitem[Sho94]{SHOUP1994371}
Victor Shoup.
\newblock Fast construction of irreducible polynomials over finite fields.
\newblock {\em Journal of Symbolic Computation}, 17(5):371--391, 1994.

\bibitem[SKL{\etalchar{+}}06]{1623822}
Minho Sung, A.~Kumar, Li~Li, Jia Wang, and Jun Xu.
\newblock Scalable and efficient data streaming algorithms for detecting common content in internet traffic.
\newblock In {\em 22nd International Conference on Data Engineering Workshops (ICDEW'06)}, pages 27--27, 2006.

\bibitem[SST{\etalchar{+}}18]{SchmidtSTTM18}
Ludwig Schmidt, Shibani Santurkar, Dimitris Tsipras, Kunal Talwar, and Aleksander Madry.
\newblock Adversarially robust generalization requires more data.
\newblock In {\em Advances in Neural Information Processing Systems 31: Annual Conference on Neural Information Processing Systems 2018, NeurIPS.}, pages 5019--5031, 2018.

\bibitem[SV11]{Eprint:SmaVer11}
N.P. Smart and F.~Vercauteren.
\newblock Fully homomorphic {SIMD} operations.
\newblock Cryptology ePrint Archive, Report 2011/133, 2011.
\newblock \url{https://eprint.iacr.org/2011/133}.

\bibitem[SZS{\etalchar{+}}14]{SzegedyZSBEGF14}
Christian Szegedy, Wojciech Zaremba, Ilya Sutskever, Joan Bruna, Dumitru Erhan, Ian~J. Goodfellow, and Rob Fergus.
\newblock Intriguing properties of neural networks.
\newblock {\em International Conference on Learning Representations}, 2014.

\bibitem[TKP{\etalchar{+}}18]{TramerKPGBM18}
Florian Tram{\`{e}}r, Alexey Kurakin, Nicolas Papernot, Ian~J. Goodfellow, Dan Boneh, and Patrick~D. McDaniel.
\newblock Ensemble adversarial training: Attacks and defenses.
\newblock In {\em 6th International Conference on Learning Representations, {ICLR}, Conference Track Proceedings}, 2018.

\bibitem[VSGB05]{Venkataraman2005NewSA}
Shobha Venkataraman, Dawn~Xiaodong Song, Phillip~B. Gibbons, and Avrim Blum.
\newblock New streaming algorithms for fast detection of superspreaders.
\newblock In {\em Network and Distributed System Security Symposium}, 2005.

\bibitem[vzGG13]{gathen_gerhard_2013}
Joachim von~zur Gathen and Jürgen Gerhard.
\newblock {\em Fast polynomial evaluation and interpolation}.
\newblock Cambridge University Press, 3 edition, 2013.

\bibitem[Waj20]{Wajc20}
David Wajc.
\newblock Rounding dynamic matchings against an adaptive adversary.
\newblock In {\em Proccedings of the 52nd Annual {ACM} {SIGACT} Symposium on Theory of Computing, {STOC}}, pages 194--207, 2020.

\bibitem[WZ12]{10.1145/2213977.2214063}
David~P. Woodruff and Qin Zhang.
\newblock Tight bounds for distributed functional monitoring.
\newblock In {\em Proceedings of the Forty-Fourth Annual ACM Symposium on Theory of Computing}, STOC '12, page 941–960, New York, NY, USA, 2012. Association for Computing Machinery.

\bibitem[WZ21]{WoodruffZ21}
David~P. Woodruff and Samson Zhou.
\newblock Tight bounds for adversarially robust streams and sliding windows via difference estimators.
\newblock In {\em 62nd {IEEE} Annual Symposium on Foundations of Computer Science, {FOCS}}, pages 1183--1196, 2021.

\bibitem[ZKWX05]{10.1145/1071690.1064258}
Qi~(George) Zhao, Abhishek Kumar, Jia Wang, and Jun~(Jim) Xu.
\newblock Data streaming algorithms for accurate and efficient measurement of traffic and flow matrices.
\newblock {\em SIGMETRICS Perform. Eval. Rev.}, 33(1):350–361, jun 2005.

\end{thebibliography}
